\documentclass[11pt]{article}

\usepackage{amsmath, amsthm, amssymb}
\usepackage{graphicx}
\usepackage{fullpage}
\usepackage[round]{natbib} 
\usepackage{float}
\usepackage{multirow}
\usepackage{indentfirst}
\usepackage[english]{babel}	
\usepackage{subfig}
\usepackage{booktabs}	
\usepackage[nottoc,numbib]{tocbibind}
\usepackage[hidelinks]{hyperref}
\usepackage[utf8]{inputenc}
\usepackage[T1]{fontenc}
\usepackage{xcolor}

\newtheorem{innercustomprop}{Proposition}
\newenvironment{customprop}[1]
  {\renewcommand\theinnercustomprop{#1}\innercustomprop}
  {\endinnercustomprop}
  
\newtheorem{innercustomlem}{Lemma}
\newenvironment{customlem}[1]
  {\renewcommand\theinnercustomlem{#1}\innercustomlem}
  {\endinnercustomlem}

\newtheorem{innercustomcor}{Corollary}
\newenvironment{customcor}[1]
  {\renewcommand\theinnercustomcor{#1}\innercustomcor}
  {\endinnercustomcor}

\newtheorem{prop}{Proposition}[section]
\newtheorem{cor}{Corollary}[section]
\newtheorem{lem}{Lemma}[section]
\DeclareMathOperator{\IF}{IF}
\DeclareMathOperator{\Cov}{Cov}
\DeclareMathOperator{\Q}{Q}
\DeclareMathOperator{\Var}{Var}
\DeclareMathOperator{\diag}{diag}

\DeclareMathOperator*{\argmin}{arg\,min}
\DeclareMathOperator*{\argmax}{arg\,max}

\providecommand{\keywords}[1]{\textbf{\textit{Key words---}} #1}
\begin{document}

\title{Robust functional regression based on principal components}
\author{Ioannis Kalogridis\thanks{Ioannis.Kalogridis@kuleuven.be} \ and Stefan Van Aelst\thanks{Stefan.VanAelst@kuleuven.be} \\ \ \\ Department of Mathematics, KU Leuven, Celestijnenlaan 200B, 3001 Leuven, Belgium}
\maketitle

\begin{abstract}
Functional data analysis is a fast evolving branch of modern statistics and the functional linear model has become popular in recent years. However, most estimation methods for this model rely on generalized least squares procedures and therefore are sensitive to atypical observations. To remedy this, we propose a two-step estimation procedure that combines robust functional principal components and robust linear regression. Moreover, we propose a transformation that reduces the curvature of the estimators and can be advantageous in many settings. For these estimators we prove Fisher-consistency at elliptical distributions and consistency under mild regularity conditions. The influence function of the estimators is investigated as well. Simulation experiments show that the proposed estimators have reasonable efficiency, protect against outlying observations, produce smooth estimates and perform well in comparison to existing approaches.
\end{abstract}
\keywords{Functional linear model, robustness, influence function, functional principal components}

\section{Introduction}

Nowadays practitioners frequently observe data that may be described by independent and identically distributed  pairs $\left(X_i, Y_i \right), i = 1, \ldots, n$ where $Y_i$ is a scalar random variable and $X_{i} (\cdot)$ is a continuous-time stochastic process defined on some compact interval $I$. The relationship between $X$ and $Y$ is often of interest and the functional regression model extends the standard multiple regression model to this case, so that
\begin{equation}
\label{eq:1}
Y_i =  m(X_i) + \epsilon_i, \quad m(X_i) = \alpha + \int_{I} X_i(t) \beta(t) dt
\end{equation}
where $\alpha \in \mathbb{R}$ and $\beta \in D$ for some function space $D$ are unknown quantities that need to be estimated from the data. The random errors $\left\{ \epsilon_i \right\}_{i=1}^n$ are assumed to be independent and identically distributed with zero mean and finite variance $\sigma^2$ and they are also assumed to be independent of the predictor curves.

The model has applications in a vast number of fields, essentially everywhere where curves, spectra or images need to be associated with scalar random variables. Examples include meteorology~\citep{ramsay2006functional},  chemometrics~\citep{ferraty2006nonparametric} and diffusion tensor imaging tractography~\citep{goldsmith2014estimator}.  More generally, the functional framework can be used to deal with ultra-high dimensional regression problems under minimal smoothness assumptions on the coefficient vector. Functional linear regression techniques have also been extended to generalized functional linear models, see e.g.\ \citep{muller2005generalized} and \citep{goldsmith2011penalized}.

Due to the model's usefulness, recent years have seen an explosion of relevant research and many novel methods have been proposed. The core of these methods is either dimensionality-reduction of the predictors or regularization of the coefficient function in the form of a roughness penalty. For an overview, one may consult the monographs of \citet{ramsay2006functional, horvath2012inference,hsing2015theoretical} and more recently, \citet{kokoszka2017introduction}. See also the recent review papers by \citet{febrero2015functional} and \citet{reiss2017methods}. 
Since many of these methods are direct extensions of classical (least squares), principal component and partial least squares procedures, a drawback is that they produce estimators that are not resistant to atypical observations. As a result, a single gross error or outlying observation may significantly affect the quality of the estimators and statistical inference based on them.  \citet{maronna2013robust} further observed that atypical observations may negatively impact the smoothness of the estimated coefficient function. Therefore, they proposed a robust alternative based on  a ridge regression-type procedure that aims to limit the impact of outlying observations while also yielding smooth estimates by penalizing the integrated squared second derivative of the coefficient function.

We take a different approach in this article and, benefiting from recent advances in robust functional data analysis, propose a robust estimator that stems from the dimensionality-reduction principle. In particular, we use functional principal components based on projection-pursuit proposed by~\citet{bali2011robust}. In combination with regression MM estimation \citep{yohai1987high}, robust functional principal components yield a computationally feasible, resistant estimator that is well-suited for the analysis of high-dimensional complex data sets. We then build upon this idea and propose a conceptually simple transformation of the estimators that improves smoothness of the estimates. For regular data the performance of the resulting estimators is comparable to popular least squares based estimates while at the same time the estimator shows good robustness in presence of contamination. 

The rest of the paper is structured as follows: Section 2 lays down the basic framework while Section 3 describes the proposed estimators in detail. Section 4 presents asymptotic results for the proposed estimators. In particular, sufficient conditions for Fisher-consistency and convergence in probability are given. The influence function is derived in Section 5 where it is also compared to the influence function of a classical estimator. Numerical experiments in Section 6 demonstrate the need for robustness and the advantages of our approach with regards to both prediction of the conditional mean and estimation of the coefficient function. The proofs of the theoretical results are relegated to the appendix.

\section{Preliminaries: definitions and notation}

The most popular setting for the functional linear model was introduced by \cite{cardot1999functional} and considers a functional random variable $X$ which is a random element in $L^2[0,1]$, the space of square integrable functions on $[0,1]$ with finite second moments. More precisely,
\begin{equation}
\label{eq:2}
X:\left(\Omega, \mathcal{A}, P \right) \times [0,1]  \to \left\{ X\left(t, \omega \right): \int_{0}^{1} \left[ X(t,\omega) \right]^2 dt < \infty \right\},
\end{equation}
and $\mathbb{E} \int_0^{1} \left[ X\left(t,\omega\right) \right]^2 dt = \int_0^{1} \mathbb{E}  \left[ X\left(t,\omega\right)\right] ^2  dt < \infty$. Provided that the process $X(t,\omega)$ is measurable with respect to the product $\sigma$-field $\mathcal{B}[0,1] \times \mathcal{A}$, $X$ may be equivalently viewed as a random element of the Hilbert space of square integrable functions, see \citet[p.~190]{hsing2015theoretical}. Both formulations are useful and will be used interchangeably henceforth, depending on convenience.

Assuming that the mean function $\mu(t) := \mathbb{E}(X)(t)$ and covariance function $\gamma(s,t) := \Cov\left(X(t), X(s) \right)$ are both continuous, the process is well-known to admit a Karhunen-Loève decomposition 
\begin{equation}
\label{eq:3}
X(t) = \mu(t) + \sum_{j=1}^{\infty}  \left\{ \int_{0}^1  \left[ X(s) - \mu(s) \right] v_j(s) ds \right\} v_j(t)  = \mu(t) + \sum_{j=1}^{\infty} \sqrt{\lambda_j }Z_j v_j(t).
\end{equation}
where $\langle x, y \rangle := \int x(t) y(t) dt$ denotes the usual $L^2$ inner product. Further, $\left\{Z _j \right\}_{j} : = \left\{ \langle X - \mu, v_j \rangle /\sqrt{\lambda_j} \right\}_{j}$ are uncorrelated random variables with zero mean and unit variance and $\left\{\left(\lambda_j, v_j \right) \right\}_{j\geq1}$ are eigenvalue-eigenfunction pairs of  $\Gamma$, the self-adjoint, Hilbert-Schmidt covariance operator of the process  $X$, defined by
\vspace{0.1cm}
\begin{equation}
\label{eq:4}
\Gamma : L^2[0,1] \to L^2[0,1]  \quad  f \in L^2[0,1] \mapsto \int_{0}^{1} \gamma(s,t) f(s) ds,
\end{equation}
which is an integral operator with $\gamma(s,t)$ as its kernel. As a consequence of Mercer's theorem, it can be shown that the sum in (\ref{eq:3}) converges in mean-square and uniformly in $\left[0,1\right]$, that is,
\begin{equation}
\label{eq:5}
\lim_{n \to \infty} \sup_{t \in [0,1]} \mathbb{E} \left[ X(t) - \mu(t) -  \sum_{i=1}^n  \sqrt{\lambda_j }Z_j v_j(t) \right]^2 = 0.
\end{equation}
See \cite{hsing2015theoretical} for more details. This theorem essentially means that, on average, $X$ can be approximated arbitrarily well by its finite-dimensional projection on the subspace spanned by the eigenfunctions of its covariance operator, uniformly on its domain.

Let us define the cross-covariance operator $\Delta(\cdot) := \mathbb{E} \langle X-\mu, \cdot \rangle (Y-\mathbb{E}(Y))$, then assuming $\lambda_1 \geq \lambda_2 \geq \ldots >0$ and using the independence of $X$ and $\epsilon$, it is easy to see that 
\begin{equation}
\label{eq:6}
\beta(t)   = \sum_{j=1}^{\infty} \frac{\Delta v_j}{\lambda_j}  v_j(t),
\end{equation}
which naturally suggests the functional principal component (FPCR) estimator 
\begin{equation}
\label{eq:7}
\widehat{\beta}(t) = \sum_{j=1}^K \frac{\Delta_n \widehat{v}_j}{\widehat{\lambda}_j} \widehat{v}_j (t),
\end{equation}
where $\Delta_n$ is now the empirical cross-moment operator and $ \widehat{\lambda}_j$ denotes the $j$th leading non-zero eigenvalue of the empirical covariance operator $\widehat{\Gamma}_n$ while $\widehat{v}_j$ denotes the corresponding eigenfunction. 

The FPCR estimator in (\ref{eq:7}) is a sieves estimator that approximates an infinite-dimensional function by its projections on a sequence of finite-dimensional subspaces. This is necessitated by the fact that since $\Gamma$ is Hilbert-Schmidt(hence compact) it has an unbounded inverse if $X$ is infinite-dimensional and therefore the coefficient function cannot be estimated directly. The dimension $K$ thus acts as a smoothing parameter and the usual trade-off between bias and variance applies. Provided that $K \to \infty$ and under some regularity conditions on the rate of decay of the eigenvalues, the process $X$ and the error $\epsilon$, the estimator has been shown to converge in probability and almost surely \citep{cardot1999functional, hall2007methodology}.	

A somewhat disconcerting feature of the FPCR estimator in (\ref{eq:7}) is that it remains quite rough even when the sample size is moderately large. This fact has motivated proposals that impose smoothness of the solution. Although this may be nominally achieved by applying a roughness penalty on either the extraction of the eigenfunctions or the least squares criterion that leads to the FPCR estimator it is the latter approach that is the most popular in the literature. In particular, \citet{ramsay2006functional} and \citet{li2007rates} discuss a Fourier basis expansion with a squared harmonic acceleration and a squared second derivative penalty respectively. The case of penalized B-spline expansions is investigated by \citet{cardot2003spline} and \citet{crambes2009smoothing} while \citet{reiss2007functional} propose a hybrid approach combining B-splines and dimensionality reduction through principal components. More recently, \citet{shin2016rkhs} have proposed a family of estimators based on reproducing kernel Hilbert spaces by generalizing work of \citet{yuan2010reproducing}. 

Both kinds of penalties can be incorporated into the robust functional principal component-robust regression framework proposed herein. While only a regression penalty is considered in detail in this paper, the methodology of Section 3.1. can also be easily extended to the case of smoothed eigenfunctions estimated by the procedure of \cite{bali2011robust}, for example.

\section{The proposed estimators}

\subsection{A robust functional principal component estimator}

Our proposal is motivated by observing from (\ref{eq:6}) that $\Delta v_j/ \lambda_j = \langle \beta, v_j \rangle$, so that an estimator for $\beta(t)$ may be obtained by estimating the scores of the coefficient function on the complete sequence $\left\{ v_j \right\}_{j \geq 1}$ of orthonormal functions. From the definition of $\Delta$ and (\ref{eq:3}), it follows that estimators of these scores can be obtained by  first centering the the responses $\left\{Y_i \right\}_{i=1}^n$ and then regressing them  on $\left\{ \langle X_i - \overline{X}, v_j \rangle \right\}_{i=1}^n$ for $ j = 1, \ldots K$. Here, $\overline{X}$ denotes the pointwise sample mean of the curves $\left\{X_i\right\}_{i=1}^n$. 

Neither functional principal components derived from the covariance operator nor least-squares regression methods are robust to anomalous observations and this remains true even if penalized estimators are used. To obtain a robust method we instead propose combining M-estimators of location for functional data, \citep{sinova2018m}, functional principal components based on projection pursuit, \citep{bali2011robust}, and MM estimators for regression \citep{yohai1987high}. We briefly review these ideas and explain their place in our proposal.

Robust estimators of univariate and multivariate location of the maximum likelihood type (M-estimators) have a well-established place in robust statistics, see e.g. \citet{huber2009robust}. In the infinite-dimensional setting robust location estimators are even more important due to the large variety of possible outlying behaviour, see \cite{hubert2015multivariate} for an extensive discussion. Recently, \cite{sinova2018m} defined M-estimators in the functional setting as
\begin{equation}
\label{eq:8}
\widehat{\mu} =  \argmin_{y\in L^2[0,1]} \sum_{i=1}^n \rho\left( \left\| X_i-y\right\| \right),
\end{equation}
where $\rho:\mathbb{R}\to \mathbb{R}_{+}$ is an even continuous, nondecreasing function satisfying $\rho(0)=0$.  \cite{sinova2018m} have shown that these estimators are well-defined, have maximal breakdown value and are consistent under suitable model assumptions. They have further supplied a fast computational algorithm that makes them well-suited for the present problem.

The definition of the loss functions in (\ref{eq:8}) allows one to consider both redescending and monotone functional M-estimators of location. We shall use the Huber family of $\rho$ functions on $\mathbb{R}_{+}$, given by
\begin{equation}
\rho_{k}(x) = \begin{cases} x^2/2 & 0\leq x \leq k; \\ 
k\left(x-k/2 \right)&  k < x, \end{cases}
\end{equation}
with $k$ a tuning parameter, because these authors have found that these estimates exhibit good performance in a wide range of models. Using an absolute loss in $L^2[0,1]$ leads to the functional median which can be computed very fast \citep{gervini2008robust}. Therefore, it serves as a starting point for the algorithm computing the functional M-estimator. 
 
The motivation for the projection pursuit idea for functional principal components comes from noticing that, as in the multivariate setting, the first eigenfunction $v_1(t)$ may be derived as the solution to the problem
\begin{equation}
\label{eq:10}
\sup_{\left\{ v \in L^2[0,1]: \left\|v \right\| = 1 \right\}} \Var\left( \langle v, X \rangle \right) = \sup_{\left\{v \in L^2[0,1]: \left\|v \right\| = 1 \right\}} \langle v, \Gamma v \rangle.
\end{equation}
The supremum of this expression is the largest eigenvalue $\lambda_1$ of $\Gamma$. Subsequent directions $v(t)$ may be obtained by imposing orthogonality conditions , i.e. by maximizing $\Var (\langle v, X \rangle)$ over the set of functions $\left\{ v \in L^2 [0,1]: \langle v, v_j \rangle = 0, 1\leq j \leq m-1 \right\}$ for $m \geq 2$. The corresponding maximal variance is now equal to the $m$th largest eigenvalue of $\Gamma$ which implies that the solutions are unique provided that the eigenvalues are distinct. Estimated eigenfunctions may be obtained by replacing $\Gamma$ with $\widehat{\Gamma}_n$ in (\ref{eq:10})  or, equivalently, by replacing the population variance with the sample variance.

Since it is well-known that the sample variance is heavily influenced by outlying observations, \citet{bali2011robust} proposed using a robust scale functional as the objective function. There are several candidates for the robust scale, but we opt for the Qn estimator \citep{rousseeuw1993alternatives}. For a sample $\left\{x_1, \ldots, x_n \right\}$ this generalized $L$-estimator is defined by
\begin{equation}
Q_n = d \left\{| x_i - x_j|; i<j \right\}_{(k)},
\end{equation}
where $d$ is a constant ensuring Fisher-consistency at the given distribution and $k$ is chosen such that the order statistic $\left\{| x_i - x_j|; i<j \right\}_{(k)}$ roughly corresponds to the first quartile of the absolute pairwise differences. $Q_n$ has a number of desirable properties which include a smooth bounded influence function for the corresponding scale functional $Q$, the highest possible breakdown value in the class of location invariant and scale equivariant functionals and a high efficiency. See \cite{rousseeuw1993alternatives} for more details.

At the population level the first projection pursuit eigenfunction based on the $Q$ scale functional can now be defined as
\begin{equation}
v_1(t) = \argmax_{\left\{v \in L^2[0,1]: \left\|v \right\| = 1 \right\}} \Q \left( \langle v, X \rangle \right).
\end{equation}
As before, we look for subsequent maximizers of $\Q \left( \langle v, X \rangle \right)$ in orthogonal directions to estimate the other eigenfunctions. This problem is a direct generalization of multivariate projection-pursuit principal components \citep{li1985projection,croux2005high} to the Hilbert space of square integrable functions. Since these functions tend to be discretized this functional generalization may be understood as principal component analysis in the presence of a very large number of variables. Sample projection-pursuit functional principal components are routinely obtained by replacing the scale functional $Q$ with its sample counterpart $Q_n$.

For simplicity it will henceforth be assumed that $\alpha = 0$ in (\ref{eq:1}) so that emphasis is placed on the coefficient function $\beta(t)$. The centered and projected observations $\langle X_i - \widehat{\mu}, \widehat{v}_j \rangle_{1 \leq j \leq K }$ along with a column of ones form the predictor matrix for the regression step. Denoting the rows of the predictor matrix by $\left\{\widehat{\mathbf{x}}_i\right\}_{i=1}^n$ for simplicity, the MM-estimator $\widehat{\boldsymbol{\beta}} := ( \widehat{\beta}_0, \boldsymbol{\widehat{\beta}}_1^{\top})^{\top}$ for $\left( \langle \beta, \mathbb{E}(X) \rangle, \langle \beta, v_1 \rangle, \ldots, \langle \beta, v_K \rangle \right)$ satisfies
\begin{equation}
\label{eq:13}
\frac{1}{n}\sum_{i=1}^n \rho^{\prime}_1 \left( \frac{Y_i - \mathbf{\widehat{x}}_i^{\top} \boldsymbol{\widehat{\beta}} }{\widehat{\sigma}_n} \right) \mathbf{\widehat{x}}_i = \mathbf{0},
\end{equation}
where  $\rho_1$ is a bounded nondecreasing even function from $\mathbb{R}$ to $\mathbb{R}_{+}$ and $\widehat{\sigma}_n$ is a robust scale that is needed to make the estimator equivariant, see \cite{maronna2006robust}. Specifically, for an initial robust, consistent and equivariant regression estimator $\boldsymbol{\widehat{\beta}}^{in}$, $\widehat{\sigma}_n$ is an M-scale implicitly defined  as the solution of 
\begin{equation}
\label{eq:14}
\frac{1}{n} \sum_{i=1}^n \rho_{0} \left( \frac{Y_i - \mathbf{\widehat{x}}_i^{\top} \boldsymbol{\widehat{\beta}}^{in}}{\widehat{\sigma}_n(\boldsymbol{\widehat{\beta}}^{in})} \right) = b,
\end{equation}
where $\rho_{0}$ is another $\rho$ function satisfying $\rho_1(x) \leq \rho_{0}(x)$. The constant $b$ controls the  breakdown value of the estimator. Let $a := \sup \rho_0(x)$,  then taking $b = 0.5 a$ ensures the maximal breakdown value of $50\%$. 

Although other initial estimators are permitted, $\boldsymbol{\widehat{\beta}}^{in}$ is usually taken to be the associated S-estimator  \citep{rousseeuw1984robust} which is the minimizer of the robust scale $\widehat{\sigma}_n(\boldsymbol{\beta})$. Hence,
\begin{equation}
\boldsymbol{\widehat{\beta}}^{in} = \argmin_{\boldsymbol{\beta} \in \mathbb{R}^{K+1}} \widehat{\sigma}_n \left( \mathbf{r}(\boldsymbol{\beta}) \right),
\end{equation}
where for any $\boldsymbol{\beta} \in \mathbb{R}^{K+1}$, the corresponding M-scale $\widehat{\sigma}_n(\boldsymbol{\beta})$ is the solution of (\ref{eq:14}). If $\rho$ is differentiable, then $\widehat{\boldsymbol{\beta}}^{in}$ is also an M-estimator provided that the scale $\widehat{\sigma}_n = \widehat{\sigma}_n (  \mathbf{r}(\boldsymbol{\widehat{\beta}}^{in}) )$ is updated simultaneously with $\boldsymbol{\widehat{\beta}}^{in}$, as noted by \cite{maronna2006robust}. An advantage of this approach is that it yields both an initial robust estimator $\boldsymbol{\widehat{\beta}}^{in}$ as well as a robust scale $\widehat{\sigma}_n$ of the residuals.

The most commonly used bounded family of rho functions is given by Tukey's bisquare family \citep{beaton1974fitting}
\begin{equation}
\label{eq:16}
\rho_c(x) = \begin{cases} 1-\left[1-(x/c)^2 \right]^3  & |x| \leq c; 
\\ 1 & |x| > c, \end{cases}
\end{equation}
where $c$ is again a tuning parameter that controls the robustness and efficiency of the estimator. The inequality $\rho_{1}(x) \leq \rho_{0}(x)$ may then be achieved by judicious choice of $c$ for $\rho_{1}(x)$ after ensuring that (\ref{eq:14}) holds. MM-estimators avoid the common trade-off between robustness and efficiency as the breakdown value is determined by the robust scale estimator $\widehat{\sigma}_n$ while the function $\rho_1$ can be tuned to achieve a desired efficiency for the estimator.

Putting everything together, the proposed robust functional principal component regression estimator (RFPCR) for the coefficient function $\beta(t)$ is given by
\begin{equation}
\label{eq:17}
\widehat{\beta}_{{\scriptscriptstyle \text{RFPCR}}}(t)  = \sum_{j=1}^K \widehat{\beta}_{1j} \widehat{v}_j (t),
\end{equation} \\
where $\widehat{\boldsymbol{\beta}}_1$ is the MM-estimator of the slopes in (\ref{eq:13}). The coefficient function $\beta$ is thus estimated by its projection on the linear space spanned by the projection-pursuit eigenfunctions.  This approximating eigen-space may be very different from the space spanned by the classical eigenfunctions of the covariance operator in the presence of atypical observations. However, for regular data these two eigen-spaces will asymptotically coincide under some assumptions on the process $X$, as will be seen in Section 4.

If the intercept $\alpha$ is not zero then it is subsumed under the MM-estimate of the intercept $\widehat{\beta}_0$. A direct estimate of $\alpha$ can then easily be obtained by  
\begin{equation}
\label{eq:18}
\widehat{\alpha} = \widehat{\beta}_0 - \int_{[0,1]} \widehat{\beta}_{{\scriptscriptstyle \text{RFPCR}}}(t) \widehat{\mu}(t) dt =  \widehat{\beta}_0 - \sum_{j=1}^K \widehat{\beta}_{1j} \int_{[0,1]} \widehat{v}_{j}(t) \widehat{\mu}(t) dt,
\end{equation}
in parallel to classical linear regression.

The estimator in (\ref{eq:17}) fulfils the need of robustness but does not yet ensure smoothness of the estimated coefficient function. Therefore, we propose an adaptation of the estimator which achieves smoothness while at the same time exhibits the same asymptotic behaviour as $\beta_{{\scriptscriptstyle \text{RFPCR}}}(t)$.

\subsection{A robust penalized functional principal component estimator}

As mentioned previously, the estimator in (\ref{eq:17}) can be quite wiggly in the presence of noise and/or contamination. For this reason it is desirable to estimate the regression coefficients in such a way that the final estimator of $\beta(t)$ exhibits more smoothness. One way to accomplish this is by incorporating a smoothness constraint in the MM estimator. Hence, the coefficients of $\beta(t)$ can be estimated by minimizing
\begin{equation}
\label{eq:19}
\frac{1}{n} \sum_{i=1}^n \rho_1 \left( \frac{Y_i - \mathbf{\widehat{x}}_i^{\top} \boldsymbol{\beta} }{\widehat{\sigma}_n} \right) + \lambda \int_{0}^{1} \left[ \beta^{\prime \prime} (t) \right]^2 dt,
\end{equation}
where $\boldsymbol{\beta} = \left(\beta_0, \boldsymbol{\beta}_1^{\top} \right)^{\top}$ with $\beta_0 \in \mathbb{R}$ and $\boldsymbol{\beta}_1 \in \mathbb{R}^{K}$. Equivalently, the criterion can be written as
\begin{equation*}
\frac{1}{n} \sum_{i=1}^n \rho_1 \left( \frac{Y_i - \beta_0 - \mathbf{\widetilde{x}}_i^{\top} \boldsymbol{\beta}_1 }{\widehat{\sigma}_n} \right) + \lambda \boldsymbol{\beta}_1  \mathbf{A} \boldsymbol{\beta}_1,
\end{equation*}
where $\mathbf{\widehat{x}}_i = (1, \mathbf{\widetilde{x}}_i^{\top})^{\top}$ and $\mathbf{A}$ is the $K\times K$ matrix of integrated products of second derivatives of the eigenfunctions, that is, $A_{ij} = \int_{0}^{1} \widehat{v}_i^{\prime \prime} (t) \widehat{v}_j^{\prime \prime} (t) dt$. The shape of the estimated coefficient function will thus depend on the value of $\lambda$. The choice $\lambda = 0$ leads to the unconstrained minimization of the loss function corresponding to the $\beta_{{\scriptscriptstyle \text{RFPCR}}}(t)$ estimator. On the other hand, as $\lambda \to \infty$ deviations of the estimated coefficient function from a straight line will be severely penalized. To the best of our knowledge, penalties on eigenfunction expansions have not been considered for the functional linear model, neither in classical nor robust approaches, because penalization of deterministic basis expansions is generally applied. 

Direct minimization of (\ref{eq:19}) may be accomplished by the strategy of \cite{maronna2013robust}. However, to avoid an additional computational burden we take a different approach in the spirit of the modified Silvapulle estimator considered by \citep{maronna2011robust}. That is,  we propose to transform the MM-estimator $\widehat{\boldsymbol{\beta}}_{1}$ into
\begin{equation}
\label{eq:20}
\widehat{\boldsymbol{\beta}}^{tr}_{1} \left(\lambda \right) = \left(\widetilde{\mathbf{X}}^{\top} \widehat{\mathbf{W}} \widetilde{\mathbf{X}} + \lambda \mathbf{A} \right)^{-1} \widetilde{\mathbf{X}}^{\top} \widehat{\mathbf{W}} \widetilde{\mathbf{X}} \widehat{\boldsymbol{\beta}}_{1},
\end{equation}
where $\mathbf{\widetilde{X}}$ is the matrix containing $\left\{\widetilde{\mathbf{x}}_i \right\}_{i=1}^n$ as its rows and $\widehat{\mathbf{W}} = \diag \left( \left\{ \widehat{w}_i  \right\}_{i=1}^n \right)$ with $\left\{ \widehat{w}_i  \right\}_{i=1}^n$ the weights corresponding to $(\widehat{\beta}_0, \widehat{\boldsymbol{\beta}}_{1})$. That is, $\widehat{w}_i = \rho_1^{\prime} \left( \widehat{u}_i/\widehat{\sigma}_n \right)/\left( \widehat{u}_i/\widehat{\sigma}_n \right) $ with $ \widehat{u}_i = Y_i - \widehat{\beta}_0 - \widetilde{\mathbf{x}}_i^{\top} \widehat{\boldsymbol{\beta}}_1$ the estimated residuals.  A benefit of this transformation is that the asymptotic behaviour of $\widehat{\boldsymbol{\beta}}_1^{tr} $ is intimately linked with the asymptotic behaviour of $\widehat{\boldsymbol{\beta}}_{1}$, as shown later. Other roughness penalties may be incorporated by replacing $\mathbf{A}$ by the corresponding penalty matrix.

The motivation for transforming the regression estimates in this way arises from ridge regression. The ridge estimator $\widehat{\boldsymbol{\beta}}_{R}$ for mean-centered $\left(\mathbf{X}, \mathbf{Y} \right)$ fulfils 

\begin{equation}
\label{eq:21}
\widehat{\boldsymbol{\beta}}_{R} = \left(\mathbf{X}^{\top} \mathbf{X} + \lambda \mathbf{I} \right)^{-1} \mathbf{X}^{\top} \mathbf{X} \widehat{\boldsymbol{\beta}}_{OLS}.
\end{equation}
In view of this, \cite{silvapulle1991robust} proposed shrinking a  monotone M-estimator instead of the least squares estimator. The resulting estimator is easy to compute and is robust with respect to outliers in the response space as it only depends on $\mathbf{Y}$ through the M-estimator. However, it remains vulnerable to leverage points. Therefore, \cite{maronna2011robust} proposed robustly centering the variables as well as using the weighting matrix produced by the MM-estimator of $\left(\mathbf{X}, \mathbf{Y}\right)$ to downweight outlying observations in the predictor space. In the ridge regression framework \cite{maronna2011robust} found this approach  effective for the case $K<n$ but impossible for $K>>n$ due to the fact that a regular MM-estimator is not well-defined in that case. 

The extension of this idea to functional principal component regression does not present such difficulty as there are at most $(n-1)$ eigenfunctions of the (sample) covariance operator and therefore at most $(n-1)$  regressors. The transformation in (\ref{eq:20}) stems from this idea but we do not center the columns of $\widetilde{\mathbf{X}}$ as the $X_i$s were centered before projecting and, under conditions given in Section 4, 
\begin{equation}
\frac{1}{n}\sum_{i=1}^n \langle X_i - \widehat{\mu}, \widehat{v}_j \rangle \xrightarrow{P} 0,
\end{equation}
for $j = 1, \ldots, K$ by the consistency of the estimated quantities, the Law of Large Numbers and Fubini's theorem since $\mathbb{E} \langle X - \mu, v_j \rangle = \langle \mathbb{E}(X-\mu), v_j \rangle = 0$.


Based on the estimator $\widehat{\boldsymbol{\beta}}^{tr}_{1}\left(\lambda \right)$ in~\eqref{eq:20} an alternative estimator of the coefficient function $\beta(t)$ is the robust functional principal component penalized regression estimator (RFPCPR) given by 
\begin{equation}
\label{eq:23}
 \widehat{\beta}_{{\scriptscriptstyle \text{RFPCPR}}}(t)  = \sum_{j=1}^K \widehat{\beta}_{1j}^{tr}(\lambda) \widehat{v}_j (t).
\end{equation}
An estimator of $\alpha$ may be obtained as in (\ref{eq:18}). Note that the estimated coefficient function $\widehat{\beta}_{{\scriptscriptstyle \text{RFPCPR}}}$ still belongs to the subspace spanned by the first $K$ robust eigenfunctions of $X$, but the estimates of the scores $\langle \beta, v_j \rangle_{j=1}^K$ have been updated to increase the smoothness of the estimator. 

In addition to $K$, the dimension of the approximating eigenspace, the penalized estimator requires the choice of the tuning parameter $\lambda$. Generally, this is an undesirable feature as optimization over both a discrete and a continuous parameter is required, which normally results in heavy computational burden. To avoid this, we outline a computationally efficient selection strategy for the penalty parameter that works well in our experience.

\subsection{Selection of the smoothing parameters $K$ and $\lambda$}

The RFPCR estimator introduced in Section 3.1  only requires the choice of the $K$, the number of principal components that is used in the regression step. This selection may be made on the basis of a robust form of k-fold cross-validation. That is, instead of the mean-squared error as a performance criterion we may use a robust dispersion measure, such as the squared $\tau$-scale \citep{yohai1988high} which may be viewed as a truncated standard deviation and combines robustness and high efficiency.  To ensure a good trade-off in that respect we recommend selecting tuning parameters that yield approximately 80\% efficiency at the Gaussian model. Although qualitatively similar to the $Qn$, the $\tau$-scale is faster to compute and thus lends itself to a fast search through the	candidate models.

A drawback of five or ten-fold cross-validation for small but complex datasets is that the number of chosen components may depend on the initial random partition of the dataset. This problem can be overcome by n-fold (leave-one-out) cross-validation. 
Let $\widehat{\boldsymbol{\beta}} =   (\widehat{\beta}_0, \widehat{\boldsymbol{\beta}}_1^{\top} )^{\top}$ denote the MM-estimator in (\ref{eq:13}) and $\widehat{\mathbf{X}}$ the matrix containing the scores on the eigenfunctions and a column of ones. The MM-estimates may be rewritten as
\begin{equation}
\label{eq:24}
\widehat{\boldsymbol{\beta}} =  \left( \widehat{\mathbf{X}}^{\top} \mathbf{\widehat{W}} \widehat{\mathbf{X}} \right)^{-1} \widehat{\mathbf{X}}^{\top} \mathbf{\widehat{W}} \mathbf{Y},
\end{equation}
for some weighting matrix $\mathbf{\widehat{W}}$ depending on the residuals $\left\{\widehat{u}_i \right\}_{i=1}^n$. Call $\widehat{Y}_{-i}$ the predicted value of $Y_i$ computed without observation $i$. Expression (\ref{eq:24}) suggests that the MM-estimator $\widehat{\boldsymbol{\beta}}$ falls into the class of linear smoothers but the weights depend on $\mathbf{Y}$, which implies that well-known computational short-cuts for leave-one-out procedures do not extend to the present case. Nevertheless, as a first-order approximation it holds that
\begin{equation}
\widehat{u}_{-i} = Y_i - \widehat{Y}_{-i} \approx \frac{\widehat{u}_i}{1-\widehat{h}_{ii}}
\end{equation}
where $\widehat{h}_{ii}$ is the ith diagonal element of the hat matrix $\widehat{\mathbf{H}} := \widehat{\mathbf{X}}  ( \widehat{\mathbf{X}}^{\top} \mathbf{\widehat{W}} \widehat{\mathbf{X}})^{-1} \widehat{\mathbf{X}}^{\top} \mathbf{\widehat{W}} $ that is obtained upon convergence. Define $\widehat{\mathbf{u}}_{-} := \left( \widehat{u}_{-1}, \ldots, \widehat{u}_{-n} \right)^{\top}$, the vector of leave-one-out residuals. These leave-one-out residuals depend on the number of regressors. We propose to select the number of components $K$ which minimizes the squared $\tau$-scale of the leave-one-out residuals $\tau^2\left(\widehat{\mathbf{u}}_{-} \right)$, as in \cite{maronna2011robust}. 

The RFPCPR estimator introduced in Section 3.2 presents the difficulty that $\lambda$ needs to be chosen in conjunction with $K$, it is in fact nested in $K$. Selection through the "double-cross" of \citep{stone1974cross} with $\tau$-scales is possible, but too time-consuming. Therefore,  it is preferable to have a direct estimate of $\lambda$. To obtain such an estimate we rewrite the problem in \eqref{eq:19} as follows. Let $\mathbf{T} \mathbf{\Lambda} \mathbf{T}^{\top}$ be the eigenvalue-eigenvector decomposition of the penalty matrix $\mathbf{A}$ in (\ref{eq:19}) with 
\begin{equation*}
\mathbf{\Lambda} = \diag \left( \lambda_1, \ldots, \lambda_s, 0, \ldots, 0 \right)
\end{equation*}
$\lambda_1 \geq \ldots \lambda_s >0$ and $\mathbf{T}$ orthogonal. Note that $\mathbf{A}$ is at least positive-semidefinite so that it does not contain negative eigenvalues. Let $\mathbf{\Lambda}_1 = \diag\left( \sqrt{\lambda_1}, \ldots, \sqrt{\lambda_s}, 1, \ldots, 1 \right) $, then we have that
\begin{equation*}
\mathbf{T} \mathbf{\Lambda} \mathbf{T}^{\top} = \mathbf{T} \mathbf{\Lambda}_1 \begin{bmatrix}
\mathbf{I}_s & \mathbf{0} \\ \mathbf{0} &\mathbf{0}  
\end{bmatrix} \mathbf{\Lambda}_1 \mathbf{T}^{\top},
\end{equation*}
By setting $\boldsymbol{\gamma}_1 = \mathbf{\Lambda}_1 \mathbf{T}^{\top} \boldsymbol{\beta}_1$ equation (\ref{eq:19}) may be rewritten as
\begin{equation}
\label{eq:26}
\frac{1}{n} \sum_{i=1}^n \rho_1 \left( \frac{Y_i - \beta_0 - \mathbf{\widetilde{x}}_i^{\top} \mathbf{T} \boldsymbol{\Lambda}_1^{-1}  \boldsymbol{\gamma}_1   }{\widehat{\sigma}_n} \right) + \lambda \boldsymbol{\gamma}_1^{\top} \begin{bmatrix}
\mathbf{I}_s & \mathbf{0} \\ \mathbf{0} &\mathbf{0}  
\end{bmatrix} \boldsymbol{\gamma}_1.
\end{equation}
By setting $\boldsymbol{\gamma}_1 = \left(\mathbf{u}^{\top}, \mathbf{b}^{\top} \right)^{\top}$ with  $\mathbf{u} \in \mathbb{R}^s, \mathbf{b} \in \mathbb{R}^{K-s}$ and splitting $\mathbf{T} \mathbf{\Lambda}_1^{-1}$ into $\left[\mathbf{Z} \  \mathbf{F} \right]$ accordingly, this can be rewritten as
\begin{equation}
\label{eq:27}
\frac{1}{n} \sum_{i=1}^n \rho_1 \left( \frac{Y_i - \beta_0 - \mathbf{\widetilde{x}}_i^{\top} \mathbf{Z} \mathbf{u} -  \mathbf{\widetilde{x}}_i^{\top} \mathbf{F} \mathbf{b}  }{\widehat{\sigma}_n} \right) + \lambda \mathbf{u}^{\top} \mathbf{u}.
\end{equation}
Note that for the case of least squares loss ($\rho_{1}(x) = \widehat{\sigma}_n x^2)$ and upon dividing by $\sigma^2$ equation (\ref{eq:27}) corresponds to the best linear unbiased prediction (BLUP) criterion for the linear mixed model
\begin{equation}
\label{eq:28}
\mathbf{Y}|\mathbf{u} \sim \mathcal{N}\left( \beta_0 \mathbf{1}_n + \widetilde{\mathbf{X}}\mathbf{Z}\mathbf{u} + \widetilde{\mathbf{X}} \mathbf{F} \mathbf{b}, \sigma^2  \right), \quad \mathbf{u} \sim \mathcal{N}\left(0, \sigma^2/\lambda \mathbf{I}_s \right).
\end{equation}
This connection has been exploited by \cite{reiss2007functional} and \cite{goldsmith2011penalized} in order to estimate $\lambda$ by maximum or restricted maximum likelihood for their penalized functional regression estimators. 

Due to the structure of the transformation in (\ref{eq:20}), ML or REML estimation of $\lambda$ yields good results in clean data but yields unsatisfactory estimates in the presence of outliers. To overcome this problem we propose a simple adjustment of the method. Let $\mathbf{w}^{*}$ denote the vector of weights corresponding to the MM-estimator of $((\mathbf{\widetilde{X}}\mathbf{Z}, \mathbf{\widetilde{X}}\mathbf{F} ), \mathbf{Y})$, i.e. the case $\lambda = 0$ in (\ref{eq:27}). A resistant estimator of the variance components may be obtained by weighing likelihood contributions with $\mathbf{w}^{*}$. As $\rho_1(x)$ is smooth but non-convex, large outliers will receive zero weight and therefore aberrant observations will not influence the estimation of the variance components. Alternatively, one could apply "hard" rejection which gives weight zero to extreme outliers and assigns weight one to the remaining observations. Although more crude, this approach tends to work well thanks to the exact fit property of MM estimators \citep{maronna2006robust}. 

With this plug-in value of $\lambda$ depending only on the number of components, the problem is again one-dimensional and the aforementioned robust cross-validation approach may be used in order to select the number of components $K$. 

\section{Asymptotic results}

\subsection{Fisher consistency}

Before examining whether the estimator converges it is important to examine whether the correct quantities are estimated, in other words, whether the estimator is Fisher-consistent. To do this, it is convenient to view estimators as functionals on the space of distribution functions equipped with the weak topology \citep{huber2009robust}. In that sense, an estimator $T(\cdot)$ is a functional applied on the empirical distribution function $F_n$ and we call this estimator Fisher-consistent for a population parameter $\theta$ if
\begin{equation}
T(F_\theta) = \theta,
\end{equation}
so that the estimator yields the correct value of the parameter when applied to the population distribution function.

In what follows we consider the RFPCR estimator, but all remarks apply to the smoothed RFPCPR estimator as well, given that these two estimators are asymptotically equivalent as will be argued later. First, we write the functional principal component regression estimator introduced in (\ref{eq:17}) as a functional of the related distribution functions. Let $F$ denote the distribution function of the error term $\epsilon$ and let $\mathbb{P}_X$ denote the distribution law (image measure) of $X$, that is, $\mathbb{P}_X(B) = \mathbb{P}(X\in B)$ for a Borel set $B$. Note that since $X$ is a Hilbertarian random variable in general this law cannot be described by a cumulative distribution function, but we can define 
\begin{equation*}
G\left(x_1, \ldots, x_k \right) := \mathbb{P}_X \left( \langle X-\mu(\mathbb{P}_X), v_1\left(\mathbb{P}_X\right) \rangle \leq x_1, \ldots, \langle X-\mu(\mathbb{P}_X), v_K\left(\mathbb{P}_X\right) \rangle \leq x_k    \right),
\end{equation*}
as the distribution function of the vector of finite-dimensional projections. With this notation the functional corresponding to the RFPCR estimator may be written as
\begin{equation}
\label{eq:29}
\beta_{{\scriptscriptstyle \text{RFPCR}}} \left(F, \mathbb{P}_X \right)(t)  = \sum_{j=1}^K \beta_{1j} \left(F, G \right) v_j \left(\mathbb{P}_X \right) (t).
\end{equation}
This functional is Fisher-consistent if $\beta_{{\scriptscriptstyle \text{RFPCR}}} \left(F, \mathbb{P}_X  \right)(t) = \beta(t)$ for $t \in[0,1]$ or equivalently, if $\beta_{1j} \left(F, G \right) = \beta_{1j}$ and $v_j \left(\mathbb{P}_X \right)(t) = v_j(t)$ for $j=1, \ldots, K$. This depends on the properties of the underlying estimators, namely M-estimators of location for functional data, projection-pursuit functional principal component estimators and MM-estimators of regression. The following three assumptions, which we discuss after the statement of Lemma 4.1, are sufficient to obtain Fisher-consistency.

\begin{enumerate}
\item[(C1)] X has a finite-dimensional Karhunen-Loève decomposition: $X(t)$ = $\mu(t) + \sum_{j=1}^K \sqrt{\lambda_j} Z_j v_j(t)$ with $ \lambda_1 > \ldots > \lambda_K >0$.
\item[(C2)] The random variables $\left\{Z_j \right\}_{j=1}^K$ are absolutely continuous and have joint density $g\left(\mathbf{x} \right)$ satisfying $g(\mathbf{x}) = h \left(||\mathbf{x}||_{E} \right)$ for $\mathbf{x} \in \mathbb{R}^K$ and some measurable function $h:\mathbb{R} \to \mathbb{R}_{+}$.
\item[(C3)] $F$ is absolutely continuous and has a density $f$ that is even, decreasing in $|x|$ and strictly decreasing in $|x|$ in a neighbourhood of zero.
\end{enumerate}
Here, $\left\| \cdot\right\|_{E}$ denotes the Euclidean norm on $\mathbb{R}^{K}$. Lemma 4.1 is an easy consequence of these assumptions and ensures that we are indeed estimating the target quantities in (\ref{eq:17}). 

\begin{lem} \label{Fisher-cons}
Let $\beta_{{\scriptscriptstyle \text{RFPCR}}} \left(F, \mathbb{P}_X \right)(t)$ be defined according to (\ref{eq:29}), assume that $(C1)-(C4)$ hold and further that $\beta$ lies in the linear subspace spanned by $\left\{v_j \right\}_{j=1}^K$. Then $\beta_{{\scriptscriptstyle \text{RFPCR}}} \left(F, \mathbb{P}_X \right)$ is Fisher-consistent, i.e. $\beta_{{\scriptscriptstyle \text{RFPCR}}} \left(F, \mathbb{P}_X \right)(t) =  \beta(t)$ for $t \in [0,1]$. 
\end{lem}


The assumption that $X$ is finite-dimensional seems restrictive and hard to justify. In view of (\ref{eq:5}) though, (C1) comes with little loss of generality and most square integrable processes can be captured sufficiently with a finite number of eigenfunctions. Even if the process $X$ is not finite-dimensional the truncation of its series at some $K$ leads to a mean-squared  error of approximation in the conditional mean that can be bounded by
\begin{equation}
\label{eq:31}
\mathbb{E}[ \mathbb{E}(Y|X(t), t \in [0,1])- \beta_0 - \sum_{j=1}^K \beta_{1j} \int_{0}^1 \left(X(t)-\mu(t)\right) v_j(t) dt ]^2 \leq ||\beta||^2 \sum_{j=K+1}^\infty \lambda_j,
\end{equation}
with $\sum_{j=K+1}^\infty \lambda_j \to 0$ as $K \to \infty$, as shown in the Appendix. This implies that the error is minimal if $K$ is chosen large enough or if $\beta$ is small in norm. 

The practical relevance of assumption (C1) is that it ensures a finite-dimensional coefficient vector, which in conjunction with (C3) implies Fisher-consistency of MM-estimators. Fisher-consistency of the eigenfunctions has been proven more generally under a condition on the distribution of the stochastic process $X$ and irrespective of its dimension \citep{bali2011robust}. This condition is that $X$ should be an elliptically distributed Hilbertarian random variable, according to the definition given in \citet{bali2009principal}. This concept is an extension of the multivariate definition of elliptically distributed random variables as non-degenerate affine transformations of spherical random variables \citep[see][Chapter~6]{maronna2006robust}. Under a finite-dimensional Karhunen-Loève expansion this definition is equivalent to (C2) and this is satisfied, for example, in the case of a Wiener process where the random variables  $\left\{Z_j \right\}_{j=1}^K$ are not only uncorrelated but also independent.

Finally, condition (C3) is common in robust estimation of the linear model and is satisfied for all commonly encountered error terms, for example, Gaussian errors. The importance of this assumption is that it ensures that the minimum of $\mathbb{E}_F [\rho \left( (Y-\mathbf{x}^{\top}\boldsymbol{\beta})/\sigma \right) ]$
is unique. Ordinarily, for this to hold true it is also required that the predictors are not concentrated in any subspace \citep{yohaitec}, but this is not an issue here since we are dealing with projected observations in orthogonal directions and therefore $\sup _{\boldsymbol{\theta} \in \mathbb{R}^{K+1}}\mathbb{P}_G \left(\boldsymbol{\theta}^{\top} \mathbf{x} = 0 \right)<1$, as required.

\subsection{Convergence in probability}

We now show that the functional principal component regression estimators $\widehat{\beta}_{{\scriptscriptstyle \text{RFPCR}}}$ and $\widehat{\beta}_{{\scriptscriptstyle \text{RFPCPR}}}$ are consistent in the $L^2$ sense. That is,
\begin{equation}
\int_{0}^1 \left( \widehat{\beta}_{{\scriptscriptstyle \text{RFPCR}}}(t) - \beta(t) \right)^2 dt \xrightarrow{P} 0 \quad \text{and} \quad  \int_{0}^1 \left( \widehat{\beta}_{{\scriptscriptstyle \text{RFPCPR}}}(t) - \beta(t) \right)^2 dt \xrightarrow{P} 0,
\end{equation}
or, in short, $|| \widehat{\beta}_{{\scriptscriptstyle \text{RFPCR}}} - \beta\|^2 \xrightarrow{P} 0 $ and $|| \widehat{\beta}_{{\scriptscriptstyle \text{RFPCPR}}} - \beta||^2 \xrightarrow{P} 0 $ with $||\cdot ||$ denoting the $L^2$ norm. This is a natural mode of convergence to consider in the present setting as the complete and separable space $L^2[0,1]$ comes with its own (semi-)metric. $L^2$ convergence does not in general imply pointwise convergence to the $L^2$ limit but a well-known result asserts that there does exist a pointwise almost-everywhere convergent subsequence. 

To facilitate the proofs we require the following additional assumption

\begin{itemize}
\item[(C4)] The process X has finite fourth moments, i.e. $\mathbb{E} \left\|X \right\|^4 < \infty$.
\end{itemize}
In view of assumption (C1), (C4) will be true if, and only if, the random variables $\left\{Z_j \right\}_{j=1}^K$ have finite fourth moments. Condition (C4) is common in the asymptotics of functional principal components, see, e.g., \citet[Chapter ~2]{horvath2012inference} and \cite{hall2007methodology}, and is used to bound lower moments. We start with the following auxiliary lemma concerning the asymptotic behaviour of the estimated M-scale $\widehat{\sigma}_n$ and the initial S-estimator $\boldsymbol{\widehat{\beta}}^{in}$, as defined in (\ref{eq:13})-(\ref{eq:14}).

\begin{lem}
Assume that conditions (C1)-(C4) hold. Call $\widehat{\boldsymbol{\beta}}^{in}$ the initial S-estimator derived from the dataset $ (\mathbf{\widehat{X}}, \mathbf{Y})$ and $\widehat{\sigma}_n$ its associated scale. Then $\widehat{\boldsymbol{\beta}}^{in} \xrightarrow{P} \boldsymbol{\beta}$ and $\widehat{\sigma}_n \xrightarrow{P} \sigma$.  
\end{lem}
With this lemma we can now prove the main result of this section. For the asymptotics of the eigenfunctions it is important to realize that sample eigenfunctions are only defined up to the usual sign ambiguity. In the following proposition we tacitly assume that the sign has been chosen correctly.

\begin{prop} Assume that conditions (C1)-(C4) are satisfied. Then $|| \widehat{\beta}_{{\scriptscriptstyle \text{RFPCR}}} - \beta || \xrightarrow{P} 0$.
\end{prop}
Convergence in probability may be strengthened to almost sure convergence under some additional assumptions. The main complication with regard to the asymptotic theory is the fact that after centering and projecting the random variables $\left\{ \langle X_i - \widehat{\mu}, \widehat{v}_j \right  \rangle \}_{i=1}^n$, $j=1, \ldots, K$ can no longer be assumed to be independent. We overcome this problem by comparing the actual estimators with their theoretical counterparts, that is, the estimators that would be obtained if $\left\{ v_j \right\}_{j=1}^K$ and $\mathbb{E}(X)$ were known. We show that the probabilistic distance between them vanishes as the sample size increases. For this, condition (C1) is instrumental as it allows us to deduce that $||\widehat{\mu}-\mathbb{E}(X) || \xrightarrow{P} 0$, instead of just the weak convergence that we would get in infinite-dimensional spaces \citep{gervini2008robust, sinova2018m}.

The result may be readily extended to $\hat{\beta}_{{\scriptscriptstyle \text{RFPCPR}}}$ under the obvious additional condition that the eigenfunctions of $\Gamma$ have finite roughness, as in \citet{silverman1996smoothed}.

\begin{itemize}
\item[(C5)] For the eigenfunctions of the covariance operator $\Gamma$ it holds that $\int_{[0,1]} [v_j^{\prime \prime}(t)]^2 dt < \infty, \ j = 1, \ldots K$.
\end{itemize}
Condition (C5) may be restated as requiring the eigenfunctions to belong to the Sobolev space $W^2_2[0,1]$ of twice differentiable functions with an absolutely continuous first derivative and square integrable second derivative. It is well-known that these functions form a dense subspace of $L^2[0,1]$. Since $\widehat{\lambda}$ is a REML estimator, it holds that $\widehat{\lambda} = O_P(1)$.  Therefore, we have that $\widehat{\lambda} = o_{P}(n)$ and we immediately obtain the following corollary.
\begin{cor}
Assume that conditions (C1)-(C5) are satisfied. Then $|| \widehat{\beta}_{{\scriptscriptstyle \text{RFPCPR}}} - \beta||^2 \xrightarrow{P} 0 $.
\end{cor}
\noindent
This corollary follows in a straightforward manner from proposition 4.1. given that $\widehat{\boldsymbol{\beta}}_1^{tr} = \mathbf{A}_n \widehat{\boldsymbol{\beta}}_1$ and $\mathbf{A}_n \xrightarrow{P}\mathbf{I}_K$. This fact means that under general conditions $\widehat{\beta}_{{\scriptscriptstyle \text{RFPCPR}}}$ is essentially a finite-sample correction to $\hat{\beta}_{{\scriptscriptstyle \text{RFPCR}}}$, whose importance diminishes as the sample size increases. This is a desirable feature of the method as the roughness of the estimated coefficient function is mainly an issue in smaller samples and this is where the smoothness correction is more needed.

Regarding the vector of estimated scores with techniques employed previously we can prove their asymptotic normality as given in the Corollary 4.2. 
\begin{cor}
Consider the vector $\widehat{\boldsymbol{\beta}} := \left( \widehat{\beta_0}, \widehat{\boldsymbol{\beta}}_1^{\top} \right)^{\top}$ from (13). Under (C1)-(C4)
\begin{equation*}
\sqrt{n}(\widehat{\boldsymbol{\beta}} - \boldsymbol{\beta} ) \xrightarrow{D} \mathcal{N}\left(0,  \sigma^2 \frac{\mathbb{E}_F \left( \rho_{1}^{\prime}(\epsilon/ \sigma)^2 \right)}{(\mathbb{E}_F \left( \rho_1^{\prime \prime}(\epsilon/\sigma))^2 \right) } \diag\left(1, \lambda_1^{-1}, \ldots \lambda_K^{-1} \right) \right)
\end{equation*}
\end{cor} 
\noindent
The corollary indicates that under the previous conditions the scores are estimated independently asymptotically. The corollary also suggests that $\sqrt{n} \left\| \hat{\beta}_{{\scriptscriptstyle \text{RFPCR}}} - \beta \right\| = O_P(1)$ so that the estimator converges at a high rate. However, for this to hold we further need to show that $\left\|\widehat{v}_j - v_j \right\| = O_P(1/\sqrt{n})$ for $j=1, \ldots K$. A stronger notion of differentiability is needed for this. While this has been established in the multivariate setting, for the functional setting to the best of our knowledge a formal proof is still lacking. Nevertheless, we conjecture that the result holds under appropriate regularity conditions.
	
\section{Influence function}

Perhaps the most popular tool in robust statistics is the influence function, introduced by \cite{hampel1974influence}. The influence function seeks to quantify the effect of infinitesimal contamination on the functional $T(F)$ corresponding to an estimator. It is defined as
\begin{equation}
\IF\left(\mathbf{x}, T, F\right) = \lim_{t \to 0^{+} } \frac{T\left(F_t\right)-T(F)}{t} = \frac{\partial}{\partial t} T\left(F_t \right) \Big\rvert_{t = 0},
\end{equation}
where $F(t) = \left(1-t\right)F + t \delta_\mathbf{x}$ is the contaminated distribution with point-mass contamination at $\mathbf{x} \in \mathbb{R}^p$. The influence function is the Gateaux derivative of the functional $T$ defined on the space of finite signed measures in the direction $\delta_x - F$ and as such exists under very general conditions \citep{hampel2011robust}. Under some regularity conditions the influence function may be used to compute the asymptotic variance of the estimator and it is also useful for diagnostic purposes since it is asymptotically equivalent to the jackknife.

From the point of view of robustness, estimators with unbounded influence function are not desirable as small contamination can lead to large distortions of the estimates. Therefore, it is important to examine whether the estimator introduced in the previous sections possesses a bounded influence function and if that is not the case, which directions of contamination are the most harmful. Since the RFPCR estimator in (\ref{eq:29}) is composed of three parts, namely the functional M-estimator of location, the functional projection-pursuit principal components and the regression MM-estimator, it is intuitively clear that its influence function will be a combination of these influence functions.

To make the notation more tractable we shall from now on assume that $\mathbb{E}(X)(t) = 0$. In case of Fisher-consistency, i.e. under (C1)-(C3), a standard calculation shows that the influence function of $\beta_{{\scriptscriptstyle \text{RFPCR}}}$ at the product measure $F \times \mathbb{P}_X$  is given by
\begin{align}
\IF\left( \left(x, y \right), \beta_{{\scriptscriptstyle \text{RFPCR}}}, \left(F,\mathbb{P}_X \right) \right) (t)  = &  \sum_{j=1}^{K} \IF\left( \left( \langle x, v_j \rangle, y \right), \beta_{1j}, \left(F, \mathbb{P}_X\right) \right) v_j  (t)  +  \sum_{j=1}^K \beta_{1j} \IF\left(x, v_j, \mathbb{P}_X \right) (t),
\end{align}
where $x$ is a point in the functional predictor space and y a point in the scalar response space. The first term on the right indeed consists of the influence function of the regression MM-estimators evaluated at the (contaminated) scores $\left\{ \langle x, v_j \rangle \right\}_{j=1}^K$ which act as regressors, while the second term is just a linear combination of the influence functions of the projection-pursuit eigenfunctions.

Denote the vector of scores $\left\{ \langle x, v_j \rangle \right\}_{j=1}^K$ by $\mathbf{x}_0$ and let $\mathbf{x} := \left(1, \mathbf{x}_0^{\top} \right)^{\top}$. The influence function of MM-estimators for symmetric error distributions can easily be derived and is given by
\begin{equation}
\label{eq:33}
\IF\left( \left(\mathbf{x}, y \right), \left(\beta_0,\boldsymbol{\beta}_1 \right), \left(F, \mathbb{P}_X \right) \right) = \frac{\sigma}{d} \rho_1^{\prime} \left(\frac{y-\beta_0-\mathbf{x}_0^{\top} \boldsymbol{\beta}_1}{\sigma} \right) \left(\mathbb{E}_G \mathbf{x} \mathbf{x}^{\top} \right)^{-1} \mathbf{x},
\end{equation}
with $d = \mathbb{E}_F (\rho_{1}^{\prime \prime} \left(\epsilon/\sigma \right))$, see \cite{maronna2006robust}. Clearly, the influence function is unbounded in $\mathbf{x}$. Hence, leverage points, i.e. large scores on the eigenfunctions can have a large effect on the estimation. However, since $\rho_{1}$ is bounded and smooth and therefore $\rho_{1}^{\prime} \to 0 $ as $|x|\to \infty $, only good leverage points may have an effect on the estimators.  Note that $\mathbb{E}_G (\mathbf{x} \mathbf{x}^{\prime})$ is diagonal with entries given by
\begin{equation}
\mathbb{E}_G  \left(\langle X, v_j \rangle \langle X, v_i \rangle \right)= \mathbb{E}_G \left( \langle \langle X, v_j \rangle X, v_i \rangle \right) = \langle \Gamma v_j, v_i \rangle = \lambda_j \delta_{ij},
\end{equation}
for $i,j >1$, where $\delta_{ij}$ is $1$ if $i=j$ and $0$ otherwise. The first diagonal element of $\mathbb{E}_G  ( \mathbf{x} \mathbf{x}^{\top} )$ is equal to $1$ while the other entries of the first row and column are zero because $X$ is centered. This implies that good leverage points in directions with small spread (small eigenvalues) have the strongest effect on the estimator.

The influence function of functional principal components based on projection pursuit was studied in \cite{bali2015influence}, who derived the influence functions of the eigenvalues and the eigenfunctions. In particular, the  IF of the $k$th eigenfunction is given by
\small
\begin{align}
\IF\left(x, v_k, \mathbb{P}_X \right)(t)  =  \sum_{j=1}^{k-1} \frac{\sqrt{\lambda}_j}{\lambda_k - \lambda_j} \IF^{\prime}\left( \frac{\langle x, v_j \rangle}{\sqrt{\lambda}_j}, Q, G \right) \langle x, v_k \rangle v_j(t) + \sum_{j\geq k+1} \frac{\sqrt{\lambda}_k}{\lambda_k - \lambda_j}\IF^{\prime}\left( \frac{\langle x, v_k \rangle}{\sqrt{\lambda}_k}, Q, G \right)  \langle x, v_j \rangle v_j(t),
\end{align}
\normalsize
and naturally depends on the influence function of the scale functional Q. For a distribution function $F$ with corresponding density $f$ the influence function of the Q functional is given by
\begin{equation}
\IF\left(x, Q, F \right) = d\frac{1/4-F(x+d^{-1}) + F(x-d^{-1})}{\int f(y+d^{-1})f(y) dy},
\end{equation}
with $d$ a calibration constant that makes the estimator Fisher-consistent at a given $F$, \citep{rousseeuw1993alternatives}. A useful property of the influence function of the functional $Q$  (and of many other robust scale functionals) is that its derivative tends to zero as $|x| \to \infty$, indicating a redescending effect for large outliers. 

With this in mind, the IF of the $k$th eigenfunction is seen to be unbounded but only for small scores on some eigenfunctions and simultaneously large scores on others. In particular, as \cite{bali2015influence} remark, observations with large absolute values of $\langle x, v_j \rangle$ in combination with small absolute values of $\langle x, v_k \rangle$ for $k<j$ 
may exert significant influence on the eigenfunctions.  The directions of unboundedness for the regression estimators and the  eigenfunctions are not necessarily identical but in practice they are very similar as more often than not both estimators are vulnerable to disproportionately large scores on certain coordinates. 

To compare the influence functions of the non-robust FPCR  and robust RFPCR estimators we consider the following example. Let $X(t) = \sqrt{\lambda_1} Z_1 v_1(t) + \sqrt{\lambda_2} Z_2 v_2(t)$ with $Z_i \sim N(0,1)$,  $v_i(t) = \sqrt{2} \sin \left( \left(i-0.5 \right) \pi t \right)$, and $\lambda_i = \pi^{-2} \left(i-0.5\right)^{-2}$ and consider $x$-values given by $ x := x_1 v_1(t) + x_2 v_2(t)$ and $y$ varying freely. The squared norm of the influence function of both estimators are plotted in Figure 1. The non-robust estimator is asymptotically equivalent to the FPC estimator of \cite{cardot1999functional}, discussed in Section 2. Therefore, these two estimators share the same influence function.
\begin{figure}[H]
\centering
\subfloat{\includegraphics[width = 0.45\linewidth]{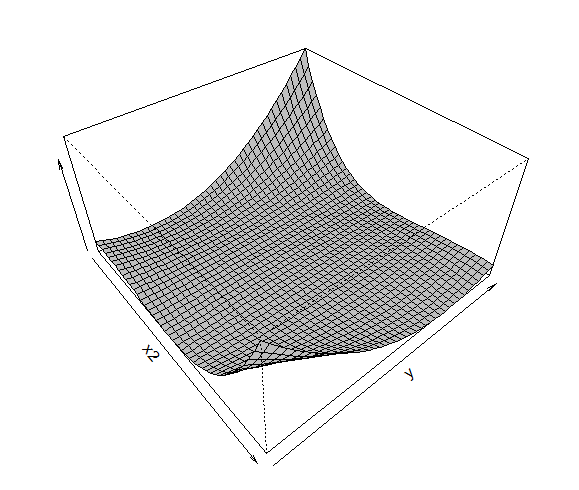}} \
\subfloat{\includegraphics[width = 0.45\linewidth]{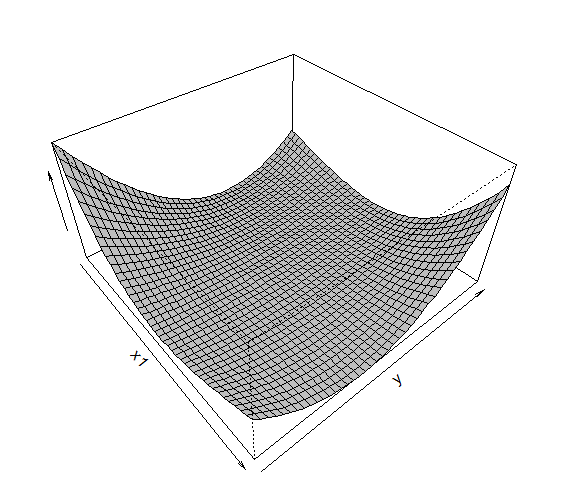}}
\\
\subfloat{\includegraphics[width=0.45\textwidth]{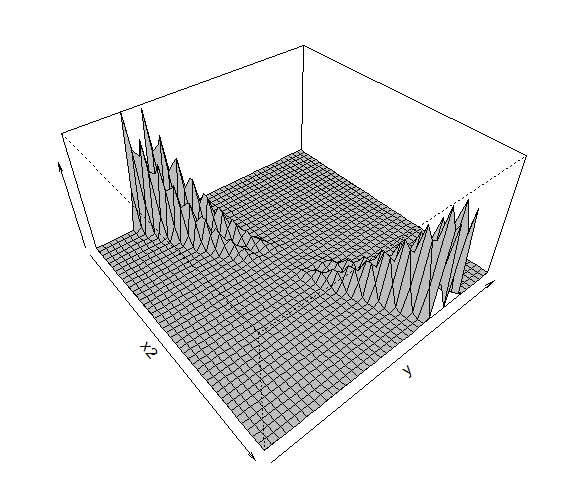}} \
\subfloat{\includegraphics[width=0.45\textwidth]{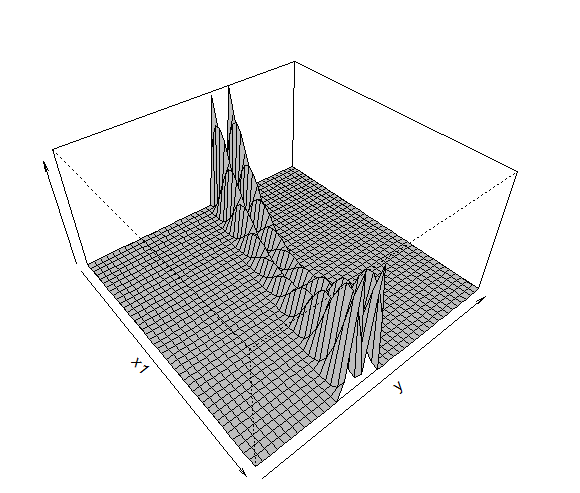}}
\caption{Norm of the IFs of the classical(top) and the robust estimator(bottom) for $x_1 = 1$ and $x_2=1$, respectively}
\end{figure}
Clearly, the classical estimator offers little protection against outlying observations. Its influence function is essentially a paraboloid and is thus unbounded in all directions. By contrast, the influence function of the robust estimator is only unbounded across the thin strips corresponding to large absolute scores on one coordinate in combination with low absolute scores on the other. In other directions one can observe the redescending effect for large outliers, which is inherited from the underlying scale functional. This represents a clear gain in robustness, which can be beneficial when analysing datasets with atypical observations.

\section{Finite sample performance}

\subsection{The competing estimators}

In this section we explore the performance of the proposed estimators by simulation and compare them to three competing methods. We first review these three competing estimators, which are the FPCR$_R$ estimator \citep{reiss2007functional}, the robust MM-spline estimator \citep{maronna2013robust} and the reproducing kernel Hilbert space estimator \citep{shin2016rkhs}.

Let $\mathbf{X}$ denote the $n \times p$ matrix of the discretized signals and let $\mathbf{B}$ denote a $p \times S$ matrix of B-spline basis functions. Then, FPCR$_R$ minimizes
\begin{equation}
\left\|\mathbf{Y}-\alpha \mathbf{1} - \mathbf{X} \mathbf{B} \widehat{\mathbf{V}}_K \boldsymbol{\beta}_1 \right\|^2_E + \lambda \ \boldsymbol{\beta}^{\top}_1 \widehat{\mathbf{V}}_K^{\top} \mathbf{P}^{\top} \mathbf{P} \widehat{\mathbf{V}}_K \boldsymbol{\beta}_1,
\end{equation}
where $\mathbf{P}^{\top}\mathbf{P} = \int b_{i}^{\prime \prime}  b_{j}^{\prime \prime} $ and $\widehat{\mathbf{V}}_K$ is the matrix of the first $K$ right singular vectors of $\mathbf{X} \mathbf{B}$. An estimator for the coefficient function is then given by $\widehat{\beta}_{\text{FPCR$_R$}} = \mathbf{B} \widehat{\boldsymbol{V}}_K \widehat{\boldsymbol{\beta}}_1$. Optimal selection of the smoothing parameters $K$ and $\lambda$ is computationally intensive. In practice, the procedure is implemented by selecting $K$ such that the explained variation of $\mathbf{X}\mathbf{B}$ is $99\%$ while $\lambda$ is estimated by restricted maximum likelihood in the manner outlined in Section 3.3. The default number of B-spline basis functions is 40.

The FPCR$_R$ estimator has been adapted to a variety of settings including functional generalized linear models. However, it is not robust to outliers. Generalizing work from \cite{crambes2009smoothing} \cite{maronna2013robust} were the first to propose a robust functional regression estimator. Their estimator $(\widehat{\alpha}_{MMSp}, \widehat{\beta}_{MMSp}(t))$ minimizes
\begin{equation}
\label{eq:40}
\widehat{\sigma}^2 \sum_{i=1}^n \rho_1 \left( \frac{Y_i - \alpha - p^{-1} \sum_{j=1}^p x_{ij} \beta(t_j)}{\widehat{\sigma}} \right) + \lambda \left(\frac{1}{p} \sum_{i=1}^p \pi_{\beta}(t_j)^2 + \int_{0}^1 \left[\beta^{\prime \prime}(t)\right]^2 dt \right),
\end{equation}
where $\pi_{\beta}$ denotes the projection of $\beta$ onto the space of linear functions and $\rho_1$ is a bounded loss function. The solution to this problem can be shown to be a cubic spline with knots at the time points. \cite{maronna2013robust} propose to select $\lambda$ based on robust leave-one-out cross-validation and to select the grid of candidate values based on the resulting effective degrees of freedom. To obtain an estimator that is robust against leverage points the authors also propose starting the iterations from an initial S-estimator that also yields an estimate $\widehat{\sigma}$, see \cite{maronna2013robust} for further details. The function $\rho_1$ is tuned for 85\% efficiency at the Gaussian model.

Let $K:[0,1]^2 \to \mathbb{R}$ denote the reproducing kernel of the Sobolev space $W^2_2[0,1]$ and define $\eta_i(t) := \int_{0}^1 x_i(u) K(u,t) du$ for $i=1, \ldots, n$. Then \cite{shin2016rkhs} propose to estimate the intercept and the coefficient function by minimizing
\begin{equation}
\label{eq:41}
\frac{1}{n} \sum_{i=1}^n \rho \left( \frac{Y_i - \alpha - \langle \eta_i, \beta \rangle_{W_2^2} }{\widehat{\sigma}} \right) + \lambda \int_{0}^1 \left[\beta^{\prime \prime}(t)\right]^2 dt,
\end{equation}
over $\alpha \in \mathbb{R}$ and $\beta \in \mathcal{W}^2_2 [0,1]$ where $\langle \cdot \rangle_{W^2_2}$ denotes the associated inner product \citep{hsing2015theoretical}. \cite{shin2016rkhs} consider both convex and non-convex loss functions $\rho$ and obtain $\widehat{\sigma}$ from an initial L1 estimator corresponding to $\rho(x) = |x|$. Following their suggestion, $\rho$ is taken to be the Tukey bisquare function and is tuned for 95\% efficiency. The penalty parameter is chosen through generalized cross-validation.

All estimators were implemented in the freeware $\texttt{R}$ \citep{R}. The FPCR$_R$ estimator is implemented through the package \texttt{refund} \citep{refund} and the remaining two estimators are implemented through custom-made functions according to the algorithms provided in the papers. The RFPCR and RFPCPR estimators were both tuned for 95\% nominal efficiency.

\subsection{Numerical results}

We are particularly interested in examining how well the estimators perform under varying levels of noise, contamination, smooth and wiggly coefficient functions and different discretizations of the curves. The latter is an important aspect of the problem since the curves are only rarely observed in their entirety and very often one has to content with noisy measurements at a finite number of points. The following two models represent the building blocks of the simulation experiments.

\textbf{Model 1} (Smooth coefficient function) The predictor curves and the coefficient function are given by
\begin{align*}
X_{ij} &= \mu(t_j) + 0.9 u_{ij}\sqrt{| \mu\left(t_j\right) |} \\ 
\beta_j & = \sqrt{t_{j}} 
\end{align*}
with $\mu(t) = \sin\left(6\pi t\right) \left(t+1\right)$, $\mathbf{u}_i \stackrel{iid}{\sim} \mathcal{N}_p \left( \mathbf{0}, \boldsymbol{\Sigma} \right)$ and $t_j = j/p, j = 1, \ldots, p$. The elements of $\boldsymbol{\Sigma}$ are given by $\Sigma_{ij} =\left( 1+\left( \frac{1}{\rho}-1 \right)(i-j)^2\right)^{-1}$. The mean function $\mu(t)$ corresponds to a sinusoid with increasing amplitude but the predictor curves are contaminated with noise proportional to the square root of the absolute value of $\mu(t)$.  The correlation structure of $\mathbf{u}_i$ indicates that the lag-one correlation between them is equal to $\rho$ and the correlations decay with some persistence. This model has been considered in \citet{maronna2013robust}.

\textbf{Model 2} (Wiggly coefficient function) The predictor curves and the coefficient function are given by
\begin{align*}
X_{ij} &=  \sum_{k=1}^{50} Z_k\sqrt{\lambda}_k v_k(t_j) \\ 
\beta\left(t_j\right) & = \log(1.5t_j^2 + 10) + \cos\left(4 \pi t_j\right),
\end{align*}
with $\left\{Z_k\right\}_{k=1}^{50} \stackrel{iid}{\sim} \mathcal{N}(0,1)$, $\lambda_k = \left((k-\frac{1}{2})^2 \pi^2 \right)^{-1}$ and $v_k(t) = \sqrt{2} \sin((k-0.5)\pi t)$. The predictor curves correspond to finite-dimensional representations of a Wiener process while the coefficient function exhibits oscillations around the logarithmic trend.  The $v_j$s are the eigenfunctions of the covariance operator of $X$, but the coefficient function is not in their linear span. Hence, it can only be approximated by its projection. A similar model was used by \cite{cardot2003spline}.  
\begin{figure}[ht!]
\centering
\subfloat{\includegraphics[width = 0.49\textwidth]{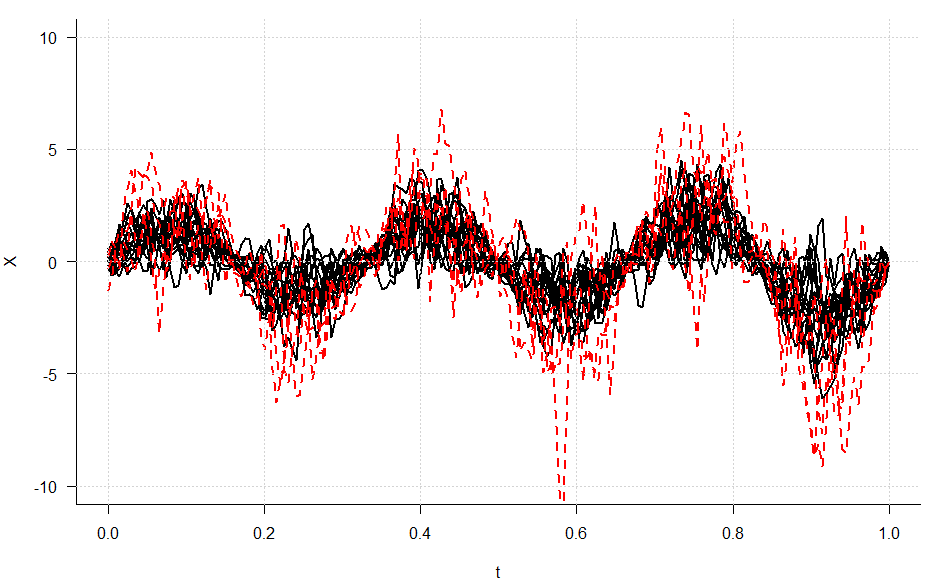}} \
\subfloat{\includegraphics[width = 0.49\textwidth]{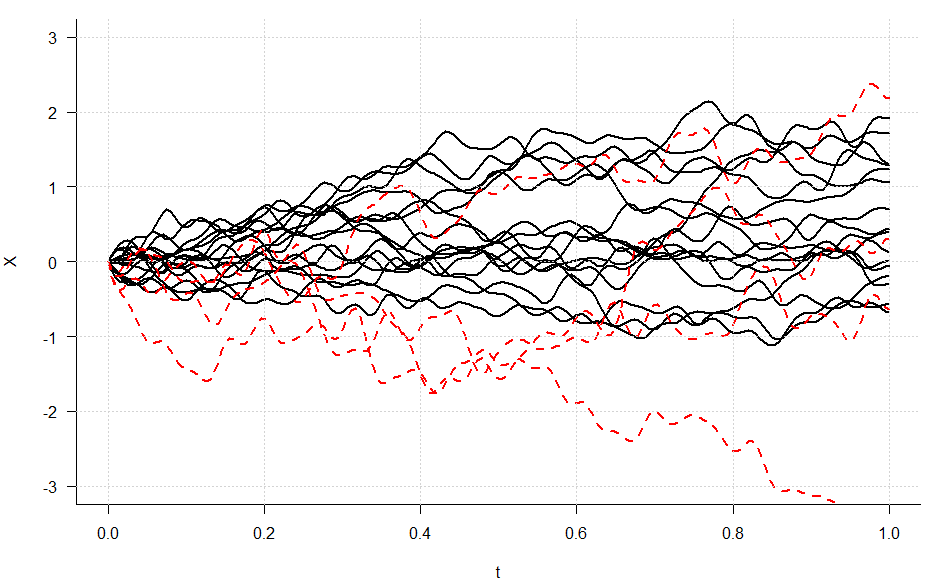}}
\caption{Clean(solid) and contaminated curves(dashed) from the two models}
\label{fig:2}
\end{figure}

In both models the response is generated according to $\mathbf{Y} = \mathbf{y}_0 + \sigma \mathbf{e}$ where $\mathbf{y}_0 = \mathbf{X} \boldsymbol{\beta}$, $\mathbf{e}$ are iid $\mathcal{N}(0,1)$ errors and $\sigma$ regulates the noise-to-signal ratio(NSR). We adopt the contamination scheme from \cite{maronna2013robust}. For $m := n \epsilon$ multiply the first $m$ rows of $\mathbf{X}$ by 2 and modify the corresponding responses by  $y_i = 2\gamma y_{0i}$. Scaling the curves affects their shape and oscillation, hence the scaled curves may be viewed as shape and amplitude outliers, \citep{hubert2015multivariate}. The constant $\gamma$ changes the relationship between predictors and response so that these observations correspond to bad leverage points.  Some clean and contaminated curves are depicted in Figure \ref{fig:2}. 

In order to represent the increasingly frequent setting of ultra-high dimensional data, we consider modest sample sizes of $n=60$ and $p= \left\{ 100, 200 \right\}$. For each of these values the NSR is set to $\left\{0.02, 0.05, 0.1 \right\}$ and we consider datasets with $0\%$, $10\%$ and $20\%$ of contamination. Several values of $\rho$ between $0.5$ and $0.9$ as well as $\gamma$ between $1.1$ and $1.8$ were considered with no qualitative differences across estimators. Hence, we only report the results for $\gamma = 1.7$ and $\rho = 0.7$. 

To compare the methods we consider both predictive and estimation evaluation criteria. Let $\boldsymbol{\widehat{\beta}} := ( \widehat{\beta}(t_1), \ldots, \widehat{\beta}(t_p) )^{\top}$ denote the point estimates. The predictive criterion is $n^{-1} ||\mathbf{y}_0 - \mathbf{X} \boldsymbol{\widehat{\beta}} (\mathbf{X^{*}}, \mathbf{y}^{*}) ||^2_{E} $ which corresponds to the mean-squared prediction error when the estimator is applied to clean data. This criterion measures how well the mean functional is predicted using the contaminated data at our disposal. The estimation criterion is $p^{-1} || \boldsymbol{\widehat{\beta}}(\mathbf{X^{*}}, \mathbf{y}^{*}) - \boldsymbol{\beta}||^2_{E}$, which is an approximation for the integrated squared error. Tables 1 and 2 display the performances of the estimators in each configuration for $1000$ replications. The best performance in each setting is highlighted in bold.

\begin{table}[ht!]

\flushleft{\textbf{Prediction}} 

\centering
{\small
\vspace{0.5cm} 
\begin{tabular}{l| c| c | c c c c c c }
p & NSR & $\epsilon$ & FPCR & FPCR$_R$ & RFPCR & RFPCPR & MMSp & RKHS$_R$   \\ \hline\hline  \\[-0.3cm]
\multirow{9}{*}{100} & \multirow{3}{*}{0.02} & 0 & 0.09 & \textbf{0.07} &   0.12 & 0.13 & 0.28 & 0.21 \\ 
& & 0.1 & 113.01 &  109.82 & 0.12 & \textbf{0.07} & 1.26  & 1.92 \\
& & 0.2 & 155.98 & 155.65 & \textbf{0.34} & 0.35 & 4.91 & 17.31 \\\cline{2-9} \\[-0.3cm]
& \multirow{3}{*}{0.05} & 0 &  0.55 &  0.28 & 0.79  & 0.32 & 1.26  & \textbf{0.25} \\
& & 0.1 & 116.51  & 110.98	  & 1.07 & \textbf{1.01} & 1.95 & 3.59  \\ 
& & 0.2 & 156.69  &  148.20 &  \textbf{3.45} & 3.58 & 7.08 & 22.63 \\\cline{2-9} \\[-0.3cm]
& \multirow{3}{*}{0.1} & 0 &  2.16 & 0.80 & 2.97 & 2.24
 & 4.46  & \textbf{0.57} \\  
& & 0.1 & 111.47 &  115.71 & 5.26 & \textbf{4.05} & 5.16 & 7.14 \\ 
& & 0.2 & 158.60  &  152.79 &  14.32 & \textbf{13.57} & 19.21 &  31.46 \\\cline{1-9} \\[-0.3cm]
\multirow{9}{*}{200} & \multirow{3}{*}{0.02} & 0 & 0.28	 & 0.18 & 0.36 & \textbf{0.16} & 1.41 & 0.77  \\ 
& & 0.1 & 530.89 & 537.69 & 0.40 & \textbf{0.26}  & 5.93 & 9.45 \\
& & 0.2 & 664.89 & 638.37 & 1.11 & \textbf{0.93}  & 48.91 & 113.39 \\\cline{2-9} \\[-0.3cm]
& \multirow{3}{*}{0.05} & 0 &  1.85 &  \textbf{0.72} & 2.50 & 1.00 & 3.45 & 0.85 \\
& & 0.1 & 541.71 & 525.95 &  3.45 & \textbf{2.25} & 9.30 & 15.59\\ 
& & 0.2 &  660.25 & 629.36 & 10.40 & \textbf{9.57} & 55.69 & 120.12 \\\cline{2-9} \\[-0.3cm]
& \multirow{3}{*}{0.1} & 0 &  7.05 & 1.92 & 8.34  & 3.13 & 9.49 & \textbf{1.63} \\  
& & 0.1 & 531.30 & 513.85 & 21.62 & \textbf{17.20} & 21.51	 & 30.16 \\ 
& & 0.2 &  664.36 &   626.49 &  74.73 & \textbf{71.73} & 83.37 & 155.82  \\\cline{1-9}
\end{tabular}}
\vspace{0.5cm} 
\flushleft{\textbf{Estimation}} 

\centering
{\small
\vspace{0.5cm} 
\begin{tabular}{l| c| c | c c c c c c }
p & NSR & $\epsilon$ & FPCR & FPCR$_R$ & RFPCR & RFPCPR & MMSp & RKHS$_R$   \\ \hline\hline  \\[-0.2cm]
\multirow{9}{*}{100} & \multirow{3}{*}{0.02} & 0 & 0.03 & 1.12 &  0.04 & \textbf{0.03} & 1.42 & 0.08 \\ 
& & 0.1 & 12.79 &   6.40 & 0.04 & \textbf{0.02} & 1.79 & 1.11 \\
& & 0.2 & 17.26 & 9.18 & 0.07 & \textbf{0.05} & 1.91 & 2.59 \\\cline{2-9} \\[-0.3cm]
& \multirow{3}{*}{0.05} & 0 &  0.19 &  1.47 & 0.26  & \textbf{0.09} & 1.73  & 0.09 \\
& & 0.1 & 12.83 & 6.45	  &  0.30 & \textbf{0.15} & 1.76 & 1.54 \\ 
& & 0.2 & 17.34  &  9.12 & 0.57 & \textbf{0.37} & 2.19  & 3.15 \\ \cline{2-9} \\[-0.3cm]
& \multirow{3}{*}{0.1} & 0 &  0.76 & 1.91 & 0.97 & 0.45 & 2.64 &  \textbf{0.13} \\  
& & 0.1 & 9.23 &  6.36 & 1.39 & \textbf{0.66} & 2.26 & 3.44 \\ 
& & 0.2 & 17.60  &  9.10 & 2.53 & \textbf{1.45} & 3.64 & 3.74  \\\cline{1-9} \\[-0.3cm]
\multirow{9}{*}{200} & \multirow{3}{*}{0.02} & 0 & 0.05 & 1.34 & 0.06 & \textbf{0.02} & 2.18 & 0.25 \\ 
& & 0.1 & 22.66 & 11.98 & 0.06 & \textbf{0.03} & 2.50 & 0.85 \\
& & 0.2 & 23.16 & 17.48 & 0.1 & \textbf{0.04} & 3.80 & 3.01 \\\cline{2-9} \\[-0.3cm]
& \multirow{3}{*}{0.05} & 0 & 0.28 &  2.19  & 0.40 & \textbf{0.13} & 2.52 & 0.22	 \\
& & 0.1 & 16.83 & 12.08 &  0.42 & \textbf{0.15} & 2.82 & 0.94 \\ 
& & 0.2 &  30.43 & 17.44 & 0.73 & \textbf{0.40} & 4.16 & 3.04 \\\cline{2-9} \\[-0.3cm]
& \multirow{3}{*}{0.1} & 0 &  1.10 & 3.06 & 1.42 & 0.36 & 3.41 & \textbf{0.27} \\  
& & 0.1 & 23.09 & 19.48 & 2.00 & \textbf{0.98} &	3.76 & 1.89	\\ 
& & 0.2 &  23.21 &   17.34 &  4.07 & \textbf{2.41} & 5.22 & 4.05 \\\cline{1-9}
\end{tabular}}
\caption{Prediction and estimation errors for Model 1, best performances in bold}
\end{table}

\begin{table}[ht!]

\flushleft{\textbf{Prediction}} 

\centering
{\small
\vspace{0.5cm} 
\begin{tabular}{l| c| c | c c c c c c }
p & NSR & $\epsilon$ & FPCR & FPCR$_R$ & RFPCR & RFPCPR & MMSp & RKHS$_R$    \\ \hline\hline  \\[-0.3cm]
\multirow{9}{*}{100} & \multirow{3}{*}{0.02} & 0 & 3.36 & \textbf{1.61} &   4.64 & 3.06 & 6.83 & 23.85 \\ 
& & 0.1 & 1059.30 &  986.46 & 4.41 & \textbf{3.60} & 6.30 & 1729.21  \\
& & 0.2 & 2266.09 & 2212.24 & 5.06 & \textbf{4.06} & 5.02 & 2710.11 \\\cline{2-9} \\[-0.3cm]	
& \multirow{3}{*}{0.05} & 0 &  10.16 & \textbf{8.10} & 18.51 & 10.38 & 38.48 & 23.97   \\
& & 0.1 & 1026.14  & 981.20	  & 17.96 & \textbf{12.36} & 40.13 & 2090.00  \\ 
& & 0.2 & 2302.56  &  2287.69 &  20.86 & \textbf{16.12} & 40.54 & 2827.05  \\\cline{2-9} \\[-0.3cm]
& \multirow{3}{*}{0.1} & 0 &  34.26 & \textbf{28.62} & 64.29 & 31.67
 & 153.65 & 31.42  \\  
& & 0.1 & 1059.59 &  1011.85 & 65.70 & \textbf{38.12} & 166.45 & 5769.05 \\ 
& & 0.2 & 2290.27  &  2273.04 & 81.60 & \textbf{51.28} & 227.32 & 3463.95  \\\cline{1-9} \\[-0.3cm]
\multirow{9}{*}{200} & \multirow{3}{*}{0.02} & 0 & 10.42	 & \textbf{6.77} & 13.28 & 11.26 & 24.58 & 93.37 \\ 
& & 0.1 & 4234.64 & 3954.69 & 12.84 & \textbf{11.24}  & 19.20 & 4735.21  \\
& & 0.2 &  9083.97 & 8942.57 & 15.51 & \textbf{13.48} & 16.77 & 8083.39 \\\cline{2-9} \\[-0.3cm]
& \multirow{3}{*}{0.05} & 0 &  41.64 &  \textbf{32.40} & 62.98 & 39.72 & 145.1 & 95.09 \\
& & 0.1 & 4189.38 & 3931.21 &  62.10 & \textbf{46.71} & 131.06 & 6884.85 \\ 
& & 0.2 &  9014.66 & 8891.86 & 77.47 & \textbf{63.30} & 163.62 & 7239.07 \\\cline{2-9} \\[-0.3cm]
& \multirow{3}{*}{0.1} & 0 &  137.44 & \textbf{118.32} & 227.84  & 127.10 & 570.90 & 131.81 \\  
& & 0.1 & 4232.73 & 3972.12 & 235.25 & \textbf{149.36} & 586.26 & 12700.84	\\ 
& & 0.2 &  8868.26 &   8721.79 &  300.84 & \textbf{204.68} & 670.98 & 9040.04 \\\cline{1-9}
\end{tabular}}

\vspace{0.5cm} 
\flushleft{\textbf{Estimation}} 

\centering
{\small
\vspace{0.5cm} 
\begin{tabular}{l| c| c | c c c c c c}
p & NSR & $\epsilon$ & FPCR & FPCR$_R$ & RFPCR & RFPCPR & MMSp & RKHS$_R$     \\ \hline\hline  \\[-0.3cm]
\multirow{9}{*}{100} & \multirow{3}{*}{0.02} & 0 & 0.37 & 0.31 &   0.64 & \textbf{0.31} & 1.49 & 0.48 \\ 
& & 0.1 & 27.52 &   4.48 & 0.57 & \textbf{0.34} & 1.76 & 42.56 \\
& & 0.2 & 22.70	 & 4.03 & 0.61 & \textbf{0.38} & 1.80 & 66.81 \\\cline{2-9} \\[-0.3cm]
& \multirow{3}{*}{0.05} & 0 &  1.07  &  0.52 & 2.33   & 0.52 & 6.36 & \textbf{0.47}   \\
& & 0.1 & 26.22 & 4.16	  &  1.83 & \textbf{0.59} & 7.31 & 68.64 \\ 
& & 0.2 & 32.99  &  4.04 & 1.70 & \textbf{0.67} & 8.26 & 79.28  \\\cline{2-9} \\[-0.3cm]
& \multirow{3}{*}{0.1} & 0 &  1.04 & 0.90 & 7.70 & 0.87 & 23.95 & \textbf{0.48} \\  
& & 0.1 & 36.74 &  4.03 & 6.24 & \textbf{0.95} & 29.11 & 128.66 \\ 
& & 0.2 & 36.96 &  3.98 & 6.37 & \textbf{1.09} & 49.87 & 87.98  \\\cline{1-9} \\[-0.3cm]
\multirow{9}{*}{200} & \multirow{3}{*}{0.02} & 0 & 0.46 & \textbf{0.27} & 0.61 & 0.28 & 1.53 & 0.47 \\ 
& & 0.1 & 32.14 & 3.39 & 0.52 & \textbf{0.31} & 1.61 & 66.16  \\
& & 0.2 & 5.93 & 3.22 & 0.57 & \textbf{0.35} & 81.82 & 89.28 \\\cline{2-9} \\[-0.3cm]
& \multirow{3}{*}{0.05} & 0 & 0.98  &  0.50 & 2.19 & 0.49 & 20.33 &  \textbf{0.46} \\
& & 0.1 & 26.31 & 3.37 &  1.71 & \textbf{0.56}& 9.10 & 85.30 \\ 
& & 0.2 & 28.10 & 3.34 & 1.66 & \textbf{0.65} & 9.23 & 107.93 \\\cline{2-9} \\[-0.3cm]
& \multirow{3}{*}{0.1} & 0 &  2.05 & 0.79 & 7.00 & 0.85 & 29.32 & \textbf{0.73} \\  
& & 0.1 & 33.31 & 3.16 & 5.79 & \textbf{0.94} &	36.22 & 186.70		\\ 
& & 0.2 &  28.60 &  2.98 &  5.67 & \textbf{1.06} & 50.26 & 112.62 \\\cline{1-9}
\end{tabular}}
\caption{Prediction and estimation errors for Model 2, best performances in bold}
\end{table}

\subsection{Discussion}

\subsubsection{Model 1}

The least squares procedures FPCR and FPCR$_R$ provide good estimates for uncontaminated data but their performance rapidly deteriorates in the presence of outliers.  FPCR$_R$ performs often twice as well as FPCR with respect to prediction, but significantly worse with respect to estimation. The reason for this is that the coefficient function is smooth and may be parsimoniously represented by a small number of basis functions. The large number of B-spline basis functions used by FPCR$_R$ is lacking in this respect as it imputes a lot more noise on the estimates. In that respect, the estimation performance of FPCR$_R$ may be substantially improved by restricting the number of basis functions but we have retained the default settings.

Among the robust estimators, RKHS$_R$ performs well in uncontaminated data and much better than the least squares estimators under contamination. However, contamination still has a considerable effect on the estimation. In presence of contamination it is vastly outperformed by both RFPCR and RFPCPR in terms of prediction error as well as estimation error. Also MMSp considerably outperforms RKHS$_R$ in terms of prediction error for contaminated data. We think that the lesser performance of RKHS$_R$ is mainly due to the $L1$ estimator that is used as a starting point in its algorithm. Since the $L1$ estimator has a zero breakdown value in random designs, this non-robust starting value may result in convergence to a "bad" local minimum of the objective function (\ref{eq:41}). On the other hand, the S-estimator used by RFPCR/RFPCPR has maximal breakdown value and thus yields a good starting point for the corresponding IRWLS iterations.

For smaller NSR, RFPCR and RFPCPR perform similarly but their difference grows in favour of RFPCPR as the NSR increases. In this case the RFPCR estimates become more wiggly, and thus smoothing becomes highly advantageous. RFPCPR performs very well with respect to prediction often coming close to FPCR$_R$ in absence of contamination, while even being substantially better with respect to estimation. Both functional principal component estimators outperform MMSp in most settings and the performance of the latter deteriorates noticeably as the discretization increases. Better results for MMSp may be obtained by considering low-rank regression splines and a more thorough search for the penalty parameter at the cost of additional computational effort.

\subsubsection{Model 2}

An interesting feature of the second simulation design is that the previously observed seasaw effect between prediction and estimation error is no longer present. In this more complex situation the large number of basis functions constitutes an advantage for the FPCR$_R$ method as it adds flexibility. Quite expectedly, FPCR$_R$ performs the best with respect to prediction in uncontaminated datasets and very well with respect to estimation, although in the latter case it is almost matched by RFPCPR and mostly outperformed by RKHS$_R$.

The RKHS$_R$ estimator exhibits good estimation performance but overall poor prediction performance. The reason is that it oversmooths the coefficient function and so its peaks and troughs are consistently missed. In general, although not often acknowledged, the performance of cross-validation methods can heavily depend on the selected grid of candidate values of the penalty parameter as well as its resolution. This means that performance can often be improved by extensive manual tuning but this is a difficult and time-consuming task, particularly when an iterative algorithm is used to obtain a solution to the problem. 

The wiggly coefficient function results in worse overall performance for the RFPCR estimator and demonstrates again the advantages of smoothing as RFPCPR exhibits markedly better performance. In absence of contamination, RFPCPR shows similar prediction performance as FPCR and both are outperformed by FPCR$_R$. Under contamination, RFPCPR outperforms all other estimators including MMSp which performs worse than in the previous design. The good performance of MMSp in \citep{maronna2013robust} was only attested for smooth coefficient functions, so the present experiment does not contradict previous findings.

\section{Example: Canadian weather data}

We illustrate the proposed penalized estimator on the well-known Canadian weather dataset, \citep{ramsay2006functional}. The dataset contains daily temperature and precipitation measurements averaged from 1960 to 1994 from 35 weather stations spread over the provinces of Canada. The response $Y$ is the log of the total annual precipitation and $X$ consists of 35 temperature curves. The goal of the analysis is to determine those months whose temperature most critically affects the yearly precipitation. Plots of the temperature curves and a histogram of the log-precipitation are given in Figure \ref{fig:3}.

\begin{figure}[H]
\centering
\subfloat{\includegraphics[scale = 0.50, width = 0.49\textwidth]{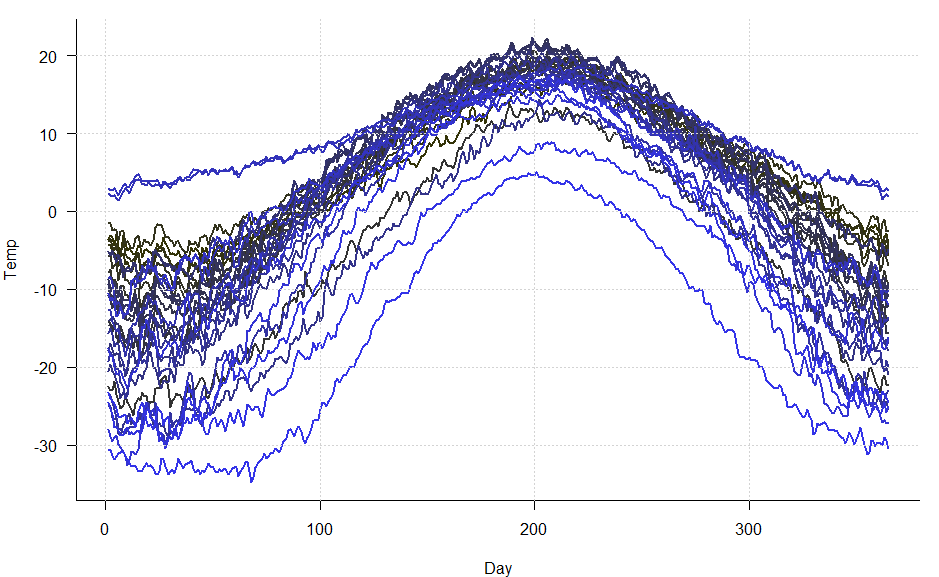}} \
\subfloat{\includegraphics[scale = 0.50, , width = 0.49\textwidth]{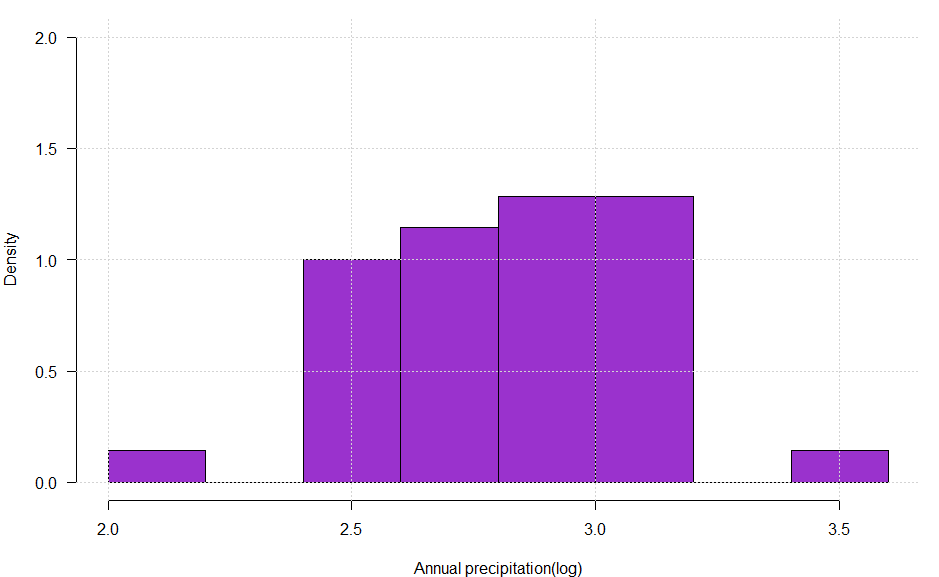}}
\caption{Temperature curves and log annual precipitation}
\label{fig:3}
\end{figure}

Interestingly, these plots already indicate the presence of outliers in both the predictor and the response spaces. These outliers correspond to atypical weather conditions within the country, for example extremely low temperatures in the Arctic regions or heavy rainfall in the province of British Columbia. Ignoring these potential outliers and estimating the coefficient function with FPCR$_R$ yields the solid curve in the left panel of Figure \ref{fig:4}. On the other hand, the solid curve in the right panel corresponds to the RFPCPR estimator of Section 3.2. While the coefficient functions are in broad agreement over the first months of the year, they substantially differ  with respect to the effect of late summer and autumn. 

\begin{figure}[H]
\centering
\subfloat{\includegraphics[scale = 0.50, width = 0.49\textwidth]{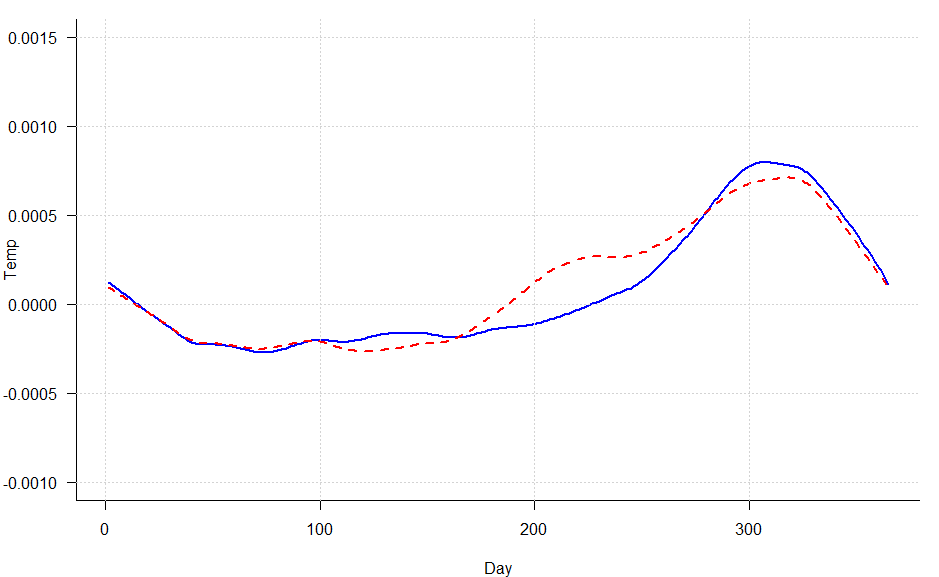}} \
\subfloat{\includegraphics[scale = 0.50, , width = 0.49\textwidth]{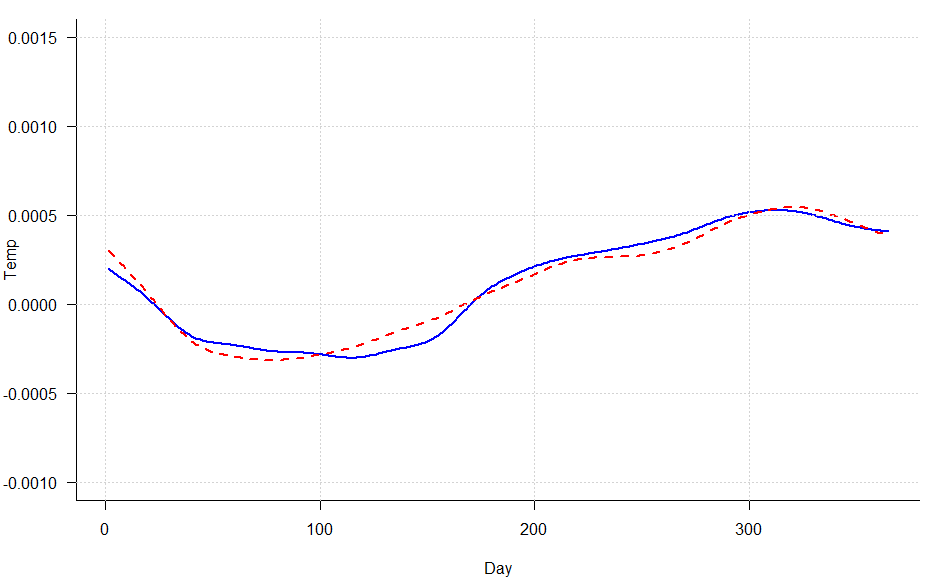}}
\caption{Estimated coefficient functions with full(solid) and outlier-free(dashed) datasets}
\label{fig:4}
\end{figure}

To determine whether this difference is attributable to outliers we may examine the residuals of the robust estimator. Examining the residuals of the FPCR$_R$ is not as informative due to the fact that least-squares estimators suffer from the masking effect \citep{rousseeuw1990unmasking}. This means that the least-squares estimator tries to fit all the data and as a result the fit is pulled towards outlying observations. On the other hand, robust estimators are little affected by atypical observations so that these observations can be identified through their large residuals. Normal QQ plots of the residuals given in Figure \ref{fig:5} illustrate this principle.

\begin{figure}[H]
\centering
\subfloat{\includegraphics[scale = 0.50, width = 0.49\textwidth]{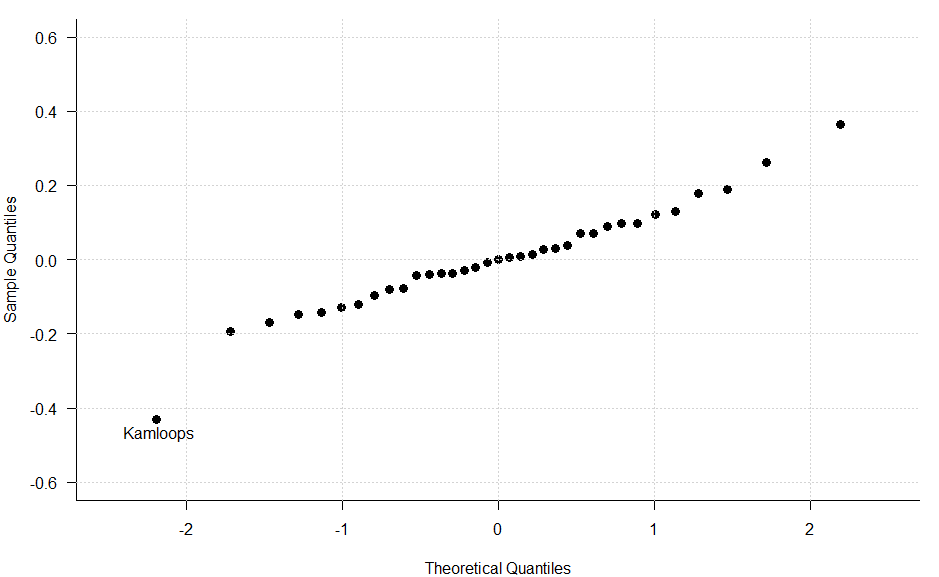}} \
\subfloat{\includegraphics[scale = 0.50, width = 0.49\textwidth]{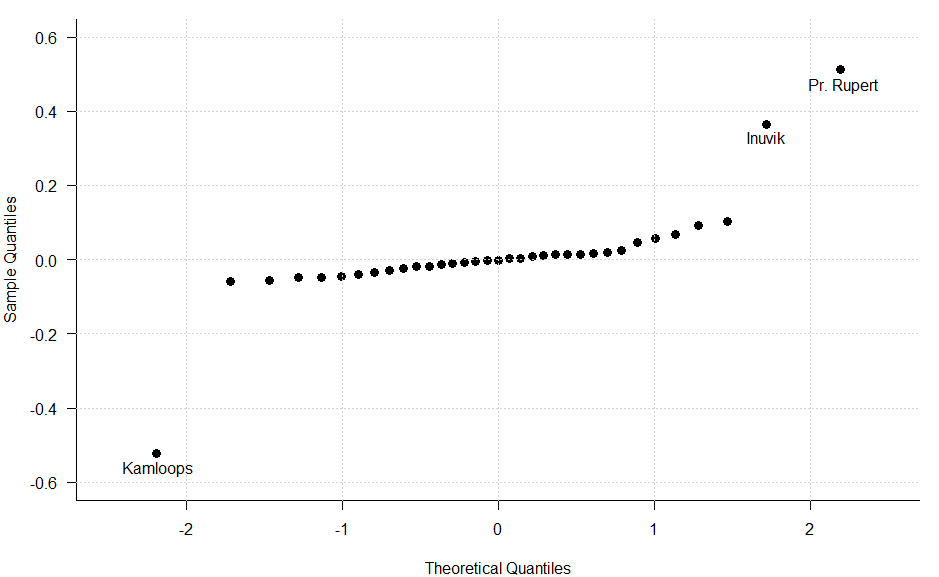}}
\caption{QQ plots of the FPCR$_R$ and the PFPCPR residuals}
\label{fig:5}
\end{figure}

The residuals of the FPCR$_R$ estimator identify only one moderately atypical observation, Kamloops. On the other hand, the residuals of the robust estimator indicate the existence of two additional outlying observations: Inuvik and Pr. Rupert. The outlyingness of Kamloops is also more pronounced. Kamloops and Prince Rupert are actually vertical outliers while Inuvik is an extreme bad leverage point with extremely low temperatures but comparatively low precipitation. Note that although Kamloops is flagged as an outlier by FPCR$_R$ this observation is not downweighted in the fitting process and as a result it still exerts influence on the estimates.  On the contrary, all three outlying observations are assigned a near-zero weight by the MM-estimator in RFPCPR.

As a sensitivity check one can remove these three outliers and recompute the estimators. This yields the dashed curves in Figure \ref{fig:4}. Omitting the outliers results in substantial changes for the FPCR$_R$ estimates, particularly during the summer and autumn months whose importance has now increased. Changes are also observed during the winter months that have now been slightly downweighted due to the exclusion of the exceedingly wet Prince Rupert. On the other hand, the robust estimates exhibit only mild adjustments. Overall, the exclusion of the outliers has brought the FPCR$_R$-estimated coefficient function closer to the robustly-estimated coefficient function.

\section{Conclusion}

In this paper we have proposed two robust functional linear regression estimators. The building blocks are robust functional principal components based on projection-pursuit and MM-estimators of regression. Regressing on the leading functional principal components is a standard recipe to obtain estimators in scalar-on-function regression, but the resulting estimators suffer from two shortcomings: non-robustness and lack of smoothness. To deal with the former we have proposed to replace classical estimates with their robust counterparts. To deal with the latter we have proposed a smoothness-improving transformation of the estimates.

For the proposed estimators we have established consistency results and also studied their robustness via the influence function. The estimators are shown to estimate the right quantities and to converge to them as the sample size increases. At the same time, their influence function reveals that the estimators are resistant to almost all kinds of infinitesimal contamination and therefore enjoy a distinct advantage over estimates obtained by minimizing an L2 norm, classical and modern alike. Simulation results have confirmed this robustness and have further shown that the estimates are not sensitive to the above-mentioned assumptions performing well in clean data and very well under contamination, even in comparison to other robust approaches. The good overall performance of the penalized estimator was also noticed in a real-data example, where it was able to detect outliers that would have been missed by least-squares procedures.

In future work we aim to examine refinements of the estimator in the direction of expressing the coefficient function in terms of a different basis and in relaxing the current set of assumptions regarding the asymptotic theory. Inclusion of scalar or other functional covariates in the current robust framework would also be off interest for both researchers and practitioners.

\newpage
\appendix

\section{Proofs of the theoretical results}

We first recall some notation. The Hilbert space $L^2[0,1]$ is the space of square integrable functions on $[0,1]$ with inner product $\langle x, y \rangle = \int_{0}^1 x(t) y(t) dt$ and norm $||x|| = \sqrt{\langle x, x \rangle  }$. The square-integrable function $X$ has a finite K-dimensional Karhunen-Loève decomposition
\begin{equation*}
X(t) = \mathbb{E}(X)(t)+  \sum_{j=1}^K \langle X -\mathbb{E}(X), v_j \rangle v_j(t).
\end{equation*}
We estimate $\mathbb{E}(X)$ by an M-estimator of location $\widehat{\mu}$  and $\left\{ v_j \right\}_{j=1}^K$ by Qn-based projection-pursuit eigenfunctions $\left\{ \widehat{v}_j \right\}_{j=1}^K$. Setting $\alpha = 0$ and inserting this decomposition in model (\ref{eq:1}) we obtain
\begin{equation*}
\mathbf{Y} = \mathbf{X} \boldsymbol{\beta} + \boldsymbol{\epsilon},
\end{equation*}
with 
\begin{equation*}
\mathbf{Y} : = \begin{bmatrix}
Y_1 \\  \vdots \\  Y_n
\end{bmatrix} \ \mathbf{X} := \begin{bmatrix}
1 & \langle X_1 -\mathbb{E}(X), v_1 \rangle  &  \ldots & \langle X_1 -\mathbb{E}(X), v_K \rangle \\ \vdots &  & \ldots & \vdots  \\ 1 &  \langle X_n-\mathbb{E}(X), v_1 \rangle  &  \ldots & \langle X_n-\mathbb{E}(X), v_K \rangle 
\end{bmatrix} \  \boldsymbol{\beta} :=  \begin{bmatrix}
\langle \beta, \mathbb{E}(X) \rangle  \\ \langle \beta, v_1 \rangle \\ \vdots \\ \langle \beta ,  v_K \rangle
\end{bmatrix} \ \boldsymbol{\epsilon} :=  \begin{bmatrix}
\epsilon_1 \\ \vdots \\  \epsilon_n
\end{bmatrix}.
\end{equation*}
In practice $\mathbf{X}$ has to be replaced by $\widehat{\mathbf{X}}$ given by
\begin{equation*}
\mathbf{\widehat{X}} = \begin{bmatrix}
1 & \langle X_1 -\widehat{\mu}, \widehat{v}_1 \rangle  &  \ldots & \langle X_1 -\widehat{\mu}, \widehat{v}_K \rangle \\ \vdots &  &\ldots & \vdots  \\ 1 &  \langle X_n-\widehat{\mu}, \widehat{v}_1 \rangle  &  \ldots & \langle X_n-\widehat{\mu}, \widehat{v}_K \rangle 
\end{bmatrix}.
\end{equation*}
An estimator for the coefficient function is obtained by 
\begin{equation*}
\widehat{\beta}_{{\scriptscriptstyle \text{RFPCR}}}(t) = \sum_{j=1}^K \widehat{\beta}_j \widehat{v}_j(t),
\end{equation*}
where $\left\{ \widehat{\beta}_{j} \right\}_{j}$ are MM-estimates of regression obtained from $(\mathbf{\widehat{X}}, \mathbf{Y} )$.

\subsection*{Derivation of (\ref{eq:31})}

\begin{proof}

Note that if $\alpha = 0$ then $\mathbb{E}(Y|X(t)| t\in [0,1] = \int X(t) \beta(t) dt$, where formally we condition on the $\sigma$-algebra generated by $\left\{X(t), t \in [0,1] \right\}$. Then
\begin{align*}
\mathbb{E}\left(Y|X(t), t\in [0,1]\right) - \beta_0 - \sum_{j=1}^K \beta_{1j} \int_{0}^1 (X(t)-\mathbb{E}(X)(t)) v_j(t) dt = \sum_{j=K+1}^{\infty} \beta_{1j} \langle X-\mu, v_j \rangle,
\end{align*}
with $\left\{\beta_{1j} \right\}_{j=1}^\infty := \left\{ \beta, v_j \rangle \right\}_{j=1}^{\infty}$. The fact that $\mathbb{E}|| X ||^2 < \infty$ implies that the covariance operator is nuclear(trace-class) with summable eigenvalues, \citep{hsing2015theoretical}. We next observe that by the Cauchy-Schwartz inequality, the Fubini-Tonelli theorem and Parseval's relation,
\begin{align*}
\mathbb{E}\left( \sum_{j=K+1}^{\infty} \beta_{1j} \langle X-\mathbb{E}(X), v_j \rangle \right)^2 &  \leq \sum_{j=K+1}^{\infty} \beta_{1j}^2  \sum_{j=K+1}^{\infty} 
\mathbb{E} \langle X-\mathbb{E}(X), v_j \rangle^2  \leq || \beta ||^2 \sum_{j=K+1}^{\infty} \lambda_j,
\end{align*}
which tends to zero as $K\to \infty$ by the summability of the series.

\end{proof}

We first prove Lemma 4.2, which establishes the consistency of the S-estimator and its associated scale.

\begin{customlem}{4.2}
Assume that conditions (C1)-(C4) hold. Call $\widehat{\boldsymbol{\beta}}^{in}$ the initial S-estimator of regression derived from the dataset $ (\mathbf{\widehat{X}}, \mathbf{Y})$ and $\widehat{\sigma}_n$ its associated scale. Then $\widehat{\boldsymbol{\beta}}^{in} \xrightarrow{P} \boldsymbol{\beta}$ and $\widehat{\sigma}_n \xrightarrow{P} \sigma$.  
\end{customlem}

\begin{proof}

By definition the estimators $\widehat{\boldsymbol{\beta}}_0^{in} $ and $\widehat{\sigma}_n$ satisfy

\begin{align*}
\frac{1}{n} \sum_{i=1}^n \rho_{0} \left( \frac{\widehat{r}_i(\widehat{\boldsymbol{\beta}}^{in})}{\widehat{\sigma}_n} \right) = b, \\
\frac{1}{n} \sum_{i=1}^n \rho_{0}^{\prime} \left( \frac{\widehat{r}_i(\widehat{\boldsymbol{\beta}}^{in})}{\widehat{\sigma}_n} \right) \widehat{\mathbf{x}}_i = \mathbf{0} ,
\end{align*}
with $\widehat{r}_i(\widehat{\boldsymbol{\beta}}^{in}) = Y_i  - \mathbf{\widehat{x}}_i^{\top} \boldsymbol{\widehat{\beta}}_0^{in}$. It is easy to see that the above implicit estimators can be rewritten in a reweighted form:
\begin{align*}
\widehat{\sigma}_n &= \frac{\widehat{\sigma}_n}{nb} \sum_{i=1}^n  \rho_0 \left(  \frac{ \widehat{r}_i \left(\widehat{\boldsymbol{\beta}}^{in} \right) }{\widehat{\sigma}_n} \right), \\
\widehat{\boldsymbol{\beta}}_0^{in} &= \left[ \frac{1}{n} \sum_{i=1}^n  \widehat{w}_i \left( \widehat{\boldsymbol{\beta}}^{in} \right)   \widehat{\mathbf{x}_i} \widehat{\mathbf{x}}_i^{\top} \right]^{-1} \frac{1}{n} \sum_{i=1}^n  \widehat{w}_i \left( \widehat{\boldsymbol{\beta}}^{in} \right) \mathbf{\widehat{x}}_i y_i,
\end{align*}
where $\widehat{w}_i(\widehat{\boldsymbol{\beta}}^{in} ) = \widehat{\sigma}_n\frac{ \rho_0 \left( \widehat{r}_i \left( \widehat{\boldsymbol{\beta}}^{in}\right)/\widehat{\sigma} \right) }{ \widehat{r}_i \left( \widehat{\boldsymbol{\beta}}^{in} \right)}$. To simplify the notation, from now on we write $\widehat{\boldsymbol{\beta}}$ and $\widehat{\sigma}$ instead of $\widehat{\boldsymbol{\beta}}^{in}$ and $\widehat{\sigma}_n$, respectively.

The estimators thus satisfy a fixed point equation $\mathbf{f}:\mathbb{R}^{K+2} \to \mathbb{R}^{K+2}$ defined for $\boldsymbol{\beta} \in \mathbb{R}^{K+1} $ and $\sigma\in \mathbb{R}$ by
\begin{equation*}
\mathbf{f}\left(\sigma , \boldsymbol{\beta}\right) = \begin{bmatrix}
\frac{\sigma}{nb} \sum_{i=1}^n   \rho_0 \left( \frac{\widehat{r}_i\left(\boldsymbol{\beta}\right)}{\sigma} \right) \\ \\
\left[ \frac{1}{n} \sum_{i=1}^n  \widehat{w}_i \left( \boldsymbol{\beta} \right)   \widehat{\mathbf{x}_i} \widehat{\mathbf{x}}_i^{\top} \right]^{-1} \frac{1}{n} \sum_{i=1}^n  \widehat{w}_i \left( \boldsymbol{\beta} \right) \mathbf{\widehat{x}}_i y_i
\end{bmatrix}.
\end{equation*}
Call $\widetilde{\sigma}$ and $\widetilde{\boldsymbol{\beta}}$ the theoretical estimates, that is, the estimates that would have been obtained when $\left\{v_j \right\}_{j=1}^K$ and $\mathbb{E}(X)$ would be known.
Since $\rho_0$ and $\rho_0^{\prime}$ are sufficiently smooth, we can use the integral form of the Mean-Value Theorem for vector-valued functions, see \cite[~Chapter 4]{ferguson2017course} and \cite{feng2014exact}, to write
\begin{equation}
\begin{bmatrix}
\widehat{\sigma} \\
\widehat{\boldsymbol{\beta}} 
\end{bmatrix} = \mathbf{f} \left(\widehat{\sigma}, \boldsymbol{\widehat{\beta}} \right) =  \mathbf{f} \left(\widetilde{\sigma}, \widetilde{\boldsymbol{\beta}} \right) + \int_{0}^1   \nabla{\mathbf{f}} \left( \widetilde{\sigma} + t \left(\widehat{\sigma} - \widetilde{\sigma} \right),  \boldsymbol{\widetilde{\beta}} + t \left( \boldsymbol{\widehat{\beta}}-\boldsymbol{\widetilde{\beta}} \right) \right) \mathrm{dt} \begin{bmatrix}
\widehat{\sigma} - \widetilde{\sigma} \\
\boldsymbol{\widehat{\beta}} - \boldsymbol{\widetilde{\beta}}
\end{bmatrix} .
\end{equation}
Upon rewriting the above equation reveals that
\begin{align}
\label{eq:43}
\begin{bmatrix}
\widehat{\sigma}-\widetilde{\sigma} \\
\boldsymbol{\widehat{\beta}}  - \boldsymbol{\widetilde{\beta}}
\end{bmatrix} &= \left[ \mathbf{I}- \int_{0}^1   \nabla{\mathbf{f}} \left( \widetilde{\sigma} + t \left(\widehat{\sigma} - \widetilde{\sigma} \right),  \boldsymbol{\widetilde{\beta}} + t \left( \boldsymbol{\widehat{\beta}}-\boldsymbol{\widetilde{\beta}} \right) \right) \mathrm{dt} \right]^{-1} \begin{bmatrix} \mathbf{f} \left(\widetilde{\sigma}, \widetilde{\boldsymbol{\beta}} \right)
-\left(\widetilde{\sigma}, \boldsymbol{\widetilde{\beta}}^{\top} \right)^{\top}\end{bmatrix} \nonumber \\
& = \left[ \int_{0}^1 \left( \mathbf{I}-    \nabla{\mathbf{f}} \left( \widetilde{\sigma} + t \left(\widehat{\sigma} - \widetilde{\sigma} \right),  \boldsymbol{\widetilde{\beta}} + t \left( \boldsymbol{\widehat{\beta}} - \boldsymbol{\widetilde{\beta}} \right) \right) \right) \mathrm{dt} \right]^{-1} \begin{bmatrix} \mathbf{f} \left(\widetilde{\sigma}, \widetilde{\boldsymbol{\beta}} \right)
-\left(\widetilde{\sigma}, \boldsymbol{\widetilde{\beta}}^{\top} \right)^{\top} \end{bmatrix}.
\end{align}
We proceed by showing that $ \begin{bmatrix} \mathbf{f} \left(\widetilde{\sigma}, \widetilde{\boldsymbol{\beta}} \right)
-\left(\widetilde{\sigma}, \boldsymbol{\widetilde{\beta}}^{\top} \right)^{\top}\end{bmatrix} = o_P(1)$ as $n \to \infty$. This holds if and only if all components are $o_{P}(1)$. For the first component, the scale, it suffices to note that
\begin{align*}
\frac{1}{n}\sum_{i=1}^n \rho_0 \left( \frac{Y_i - \widehat{\mathbf{x}}_i^{\top} \widetilde{\boldsymbol{\beta}} }{\widetilde{\sigma}} \right) = \frac{1}{n} \sum_{i=1}^n \rho_0 \left( \frac{Y_i - \mathbf{x}_i^{\top} \widetilde{\boldsymbol{\beta}}}{\widetilde{\sigma}}  \right) + \frac{1}{n\widetilde{\sigma}} \sum_{i=1}^n  \rho^{\prime}_0\left(r_i^{*} \right) \left(  \mathbf{x}_i^{\top} \widetilde{\boldsymbol{\beta}} -  \widehat{\mathbf{x}}_i^{\top} \widetilde{\boldsymbol{\beta}} \right),
\end{align*}
with $r^{*}_i$ between $\widehat{\mathbf{x}}_i^{\top} \widetilde{\boldsymbol{\beta}}/\widetilde{\sigma}$ and $\mathbf{x}_i^{\top} \widetilde{\boldsymbol{\beta}}/\widetilde{\sigma}$. Since by Theorems 4.1 and 4.2 of \citep{yohaitec} $\widetilde{\sigma}$ and $\widetilde{\boldsymbol{\beta}}$ are consistent they are also bounded in probability. 
Furthermore, 
\begin{equation*}
\left| \frac{1}{n} \sum_{i=1}^n  \rho^{\prime}_0\left(r_i^{*} \right) \left(  \mathbf{x}_i^{\top} \widetilde{\boldsymbol{\beta}} -  \widehat{\mathbf{x}}_i^{\top} \widetilde{\boldsymbol{\beta}} \right) \right| \leq \frac{C ||\boldsymbol{\widetilde{\beta}} ||}{n} \sum_{i=1}^n \left\| \mathbf{x}_i - \widehat{\mathbf{x}}_i \right\|_{E},
\end{equation*}
for some constant $C>0$ by the boundedness of $\rho_{0}^{\prime}(x)$ and the Schwarz inequality. An easy calculation now shows that
\begin{equation*}
\frac{1}{n} \sum_{i=1}^n \left\|\mathbf{\widehat{x}}_i - \mathbf{x}_i  \right\|_{E} \leq \frac{1}{n} \sum_{i=1}^n  \left\| X_i -\mathbb{E}(X) \right\| \sum_{j=1}^K \left\| \widehat{v}_j -v_j \right\| + K \left\|\widehat{\mu}-\mathbb{E}(X) \right\|,
\end{equation*}
which is $o_{P}(1)$ by the Law of Large Numbers and  the consistency of the eigenfunctions and the functional M-estimator of location, see theorem 4.2 of \cite{bali2011robust} and theorem 3.4 of \cite{sinova2018m}, respectively. Combining the above facts, it now follows that
\begin{equation}
\frac{1}{n}\sum_{i=1}^n \rho_0 \left( \frac{Y_i - \widehat{\mathbf{x}}_i^{\top} \widetilde{\boldsymbol{\beta}} }{\widetilde{\sigma}} \right) = b + o_P(1),
\end{equation}
so that the first component of $\begin{bmatrix} \mathbf{f} \left(\widetilde{\sigma}, \widetilde{\boldsymbol{\beta}} \right)
-\left(\widetilde{\sigma}, \boldsymbol{\widetilde{\beta}}^{\top} \right)^{\top}\end{bmatrix}$ is indeed $o_P(1)$. 

For the other $K+1$ components it is helpful to note that the theoretical estimator $\boldsymbol{\widehat{\beta}}$ can also be written in a weighted least squares form, so that we have
\begin{equation*}
\boldsymbol{\widetilde{\beta}} = \left[ \frac{1}{n} \sum_{i=1}^n  w_i ( \boldsymbol{\widetilde{\beta}} ) \mathbf{x}_i \mathbf{x}_i^{\top} \right]^{-1} \frac{1}{n} \sum_{i=1}^n w_i ( \boldsymbol{\widetilde{\beta}}) \mathbf{x}_i y_i,
\end{equation*}
with $w_i ( \boldsymbol{\widetilde{\beta}} ) = \widetilde{\sigma}\frac{ \rho_0 \left( r_i \left( \widetilde{\boldsymbol{\beta}} \right)/\widetilde{\sigma} \right) }{ r_i \left( \widetilde{\boldsymbol{\beta}} \right)}$ and $r_i(\boldsymbol{\beta}) = Y_i - \mathbf{x}^{\top}_i \boldsymbol{\beta}$. We compare $n^{-1}\sum_{i=1}^n \widehat{w}_i ( \boldsymbol{\widetilde{\beta}}) \widehat{\mathbf{x}}_i y_i$ and its theoretical counterpart $n^{-1} \sum_{i=1}^n w_i ( \boldsymbol{\widetilde{\beta}}) \mathbf{x}_i y_i $ and show that they  have the same limit, which by the consistency of MM-estimators is $ \sigma \mathbb{E}(\mathbf{x}Y) \mathbb{E}\rho_{0}^{\prime}(\epsilon/\sigma)/\epsilon<\infty$. Since
\begin{align*}
\left\| \frac{1}{n} \sum_{i=1}^n y_i \left(\widehat{w}_i ( \boldsymbol{\widetilde{\beta}}) \mathbf{\widehat{x}}_i - w_i ( \boldsymbol{\widetilde{\beta}}) \mathbf{x}_i  \right) \right\|_{E} 
\leq \frac{1}{n} \sum_{i=1}^n \left|y_i \right| \left\| \mathbf{\widehat{x}}_i - \mathbf{x}_i \right\|_{E} + \frac{1}{n} \sum_{i=1}^n \left| y_i \right| \left|  \right| \left\| \mathbf{x}_i \right\|_{E} \left| \widehat{w}_i ( \boldsymbol{\widetilde{\beta}}) - w_i ( \boldsymbol{\widetilde{\beta}}) \right|,
\end{align*}
the previous argument shows that the first term is $o_{P}(1)$. For the second term note that the function $\rho_0(x)/x$ has a bounded derivative and hence it satisfies a Lipschitz condition. Therefore, for some $D>0$
\begin{equation}
\left| \widehat{w}_i ( \boldsymbol{\widetilde{\beta}}) - w_i ( \boldsymbol{\widetilde{\beta}}) \right| \leq \frac{D}{\widehat{\sigma}} \left| \widehat{x}_{i}^{\prime}\boldsymbol{\widetilde{\beta}} -\mathbf{x}_i ^{\prime} \boldsymbol{\widetilde{\beta}} \right| \leq \frac{D}{\widehat{\sigma}}  \left\| \mathbf{\widehat{x}}_i - \mathbf{x}_i \right\|_{E} || \boldsymbol{ \widetilde{\beta}} ||_{E},
\end{equation}
which implies that the second term is also $o_{P}(1)$. Similarly, but more tediously, using condition (C4) we can show that
\begin{equation*}
\frac{1}{n} \sum_{i=1}^n \left\| \widehat{w}_i ( \boldsymbol{\widetilde{\beta}} ) \mathbf{\widehat{x}}_i \mathbf{\widehat{x}}_i^{\top} - w_i ( \boldsymbol{\widetilde{\beta}} ) \mathbf{x}_i \mathbf{x}_i^{\top}  \right\|_{E} = o_{P}(1),
\end{equation*}
which in combination with Slutsky's theorem now implies the desired result as
\begin{equation}
\frac{1}{n}\sum_{i=1}^n w_i ( \boldsymbol{\widetilde{\beta}} ) \mathbf{x}_i \mathbf{x}_i^{\top} \xrightarrow{P} \mathbb{E} ( \mathbf{x}\mathbf{x}^{\top} ) \sigma \mathbb{E} \left( \frac{ \rho_0^{\prime} (\epsilon/ \sigma)}{\epsilon} \right)< \infty,
\end{equation}
by the independence of $\left\{X(t); t\in [0,1] \right\}$ and $\epsilon$ and the fact that $g(x):= \langle x, y  \rangle, y \in L^2[0,1] $ is a Borel-measurable function.

To complete the proof we show that 
\begin{equation}
\int_{0}^1 \left( \mathbf{I}-    \nabla{\mathbf{f}} \left( \widetilde{\sigma} + t \left(\widehat{\sigma} - \widetilde{\sigma} \right),  \boldsymbol{\widetilde{\beta}} + t \left( \boldsymbol{\widehat{\beta}} - \boldsymbol{\widetilde{\beta}} \right) \right) \right) \mathrm{dt} = O_{P}(1).
\end{equation}
Similar calculations as in \cite{salibian2002bootstrapping} show that
\begin{equation*}
\nabla \mathbf{f}(\sigma, \boldsymbol{\beta}) = \begin{bmatrix}
\frac{1}{nb}\sum_{i=1}^n \rho_0\left(\frac{\widehat{r}_i(\boldsymbol{\beta})}{\sigma} \right) - \frac{1}{n \sigma b}\sum_{i=1}^n \rho_0^{\prime} \left( \frac{\widehat{r}_i(\boldsymbol{\beta})}{\sigma}\right)  \widehat{r}_i(\boldsymbol{\beta}) & - \frac{1}{nb} \sum_{i=1}^n \rho_{0}^{\prime} \left( \frac{\widehat{r}_i(\boldsymbol{\beta}}{\sigma}\right)\mathbf{\widehat{x}}_i^{\top} \\ \mathbf{d} & \mathbf{I}_{K+1}-\mathbf{A}^{-1} \frac{1}{n} \sum_{i=1}^n \rho_0^{\prime \prime}\left( \frac{\widehat{r}_i\left(\boldsymbol{\beta} \right)}{\sigma} \right) \frac{\mathbf{\widehat{x}}_i \mathbf{\widehat{x}}_i^{\top}}{\sigma}
\end{bmatrix},
\end{equation*}
with 
\begin{equation*}
\mathbf{A} = \frac{1}{n} \sum_{i=1}^n \frac{ \rho_0^{\prime}\left(\frac{\widehat{r}_i(\boldsymbol{\beta}}{\sigma} \right) }{\widehat{r}_i(\boldsymbol{\beta})} \mathbf{\widehat{x}}_i \mathbf{\widehat{x}}_i^{\top},
 \end{equation*}
and
\begin{align*}
\mathbf{d} = &  \mathbf{A}^{-1}  \frac{1}{n}\sum_{i=1}^n \rho_0 ^{\prime \prime} \left( \frac{\widehat{r}_i(\boldsymbol{\beta})}{\sigma}\right) \frac{\mathbf{\widehat{x}}_i \mathbf{\widehat{x}}_i}{\sigma^2}  \mathbf{A}^{-1} \frac{1}{n}\sum_{i=1}^n  \frac{\rho_0^{\prime} \left(\frac{\widehat{r}_i(\boldsymbol{\beta})}{\sigma} \right)} {\widehat{r}_i(\boldsymbol{\beta})}\mathbf{\widehat{x}}_i y_i  -  \mathbf{A}^{-1}  \frac{1}{n} \sum_{i=1}^n \rho_0^{\prime \prime} \left( \frac{r_i(\boldsymbol{\beta})}{\sigma} \right) \frac{\mathbf{\widehat{x}}_i y_i}{\sigma^2}.
\end{align*}
Therefore,
\begin{equation*}
\mathbf{I} - \nabla \mathbf{f}(\sigma, \boldsymbol{\beta}) = \begin{bmatrix}
1+ \frac{1}{n \sigma b}\sum_{i=1}^n \rho_0^{\prime} \left( \frac{\widehat{r}_i(\boldsymbol{\beta})}{\sigma}\right)  \widehat{r}_i(\boldsymbol{\beta}) - \frac{1}{nb}\sum_{i=1}^n \rho_0\left(\frac{\widehat{r}_i(\boldsymbol{\beta})}{\sigma} \right)  
& \frac{1}{nb} \sum_{i=1}^n \rho_{0}^{\prime} \left( \frac{\widehat{r}_i(\boldsymbol{\beta}}{\sigma}\right)\mathbf{\widehat{x}}_i^{\top}\\ - \mathbf{d} & \mathbf{A}^{-1} \frac{1}{n} \sum_{i=1}^n \rho_0^{\prime \prime}\left( \frac{\widehat{r}_i\left(\boldsymbol{\beta} \right)}{\sigma} \right) \frac{\mathbf{\widehat{x}}_i \mathbf{\widehat{x}}_i^{\top}}{\sigma}\end{bmatrix},
\end{equation*}
and from the above it is clear these expressions are well-defined. By the Law of Large Numbers, the smoothness and boundedness of $\rho_0^{\prime \prime}(x)$, $\rho_0^{\prime}(x) x$ and $\rho_0^{\prime}(x)/x$ as well as the asymptotic equivalence of $n^{-1} \sum_{i=1}^n \mathbf{\widehat{x}}_i \mathbf{\widehat{x}}_i^{\top}$ and $n^{-1} \sum_{i=1}^n \mathbf{x}_i \mathbf{x}_i^{\top}$ they are also $O_{P}(1)$. 

Hence, $\mathbf{I} - \nabla \mathbf{f}(\sigma, \boldsymbol{\beta})$ is bounded in probability for all $\left( \sigma, \boldsymbol{\beta}\right)$ and as a result
\begin{equation}
\int_{0}^1 \left[ \mathbf{I}-    \nabla{\mathbf{f}} \left( \widetilde{\sigma} + t \left(\widehat{\sigma} - \widetilde{\sigma} \right),  \boldsymbol{\widetilde{\beta}} + t \left( \boldsymbol{\widehat{\beta}} - \boldsymbol{\widetilde{\beta}} \right) \right) \right] \mathrm{dt} = O_{P}(1),
\end{equation}
as required.

\end{proof}

\begin{customprop}{4.1}  Assume that conditions (C1)-(C4) are satisfied. Then $\left\| \widehat{\beta}_{{\scriptscriptstyle \text{RFPCR}}} - \beta \right\| \xrightarrow{P} 0$.
\end{customprop}

\begin{proof}
We have
\begin{align*}
\left\| \widehat{\beta}_{{\scriptscriptstyle \text{RFPCR}}} - \beta \right\|  = \left\| \sum_{j=1}^K \widehat{\beta}_{j} \widehat{v}_{j} - \sum_{j=1}^K \beta_j v_j  \right\| \nonumber  & \leq \left\| \sum_{j=1}^K \left( \widehat{\beta}_{j} - \beta_j \right) \widehat{v}_{j} \right\| + \left\| \sum_{j=1}^K \beta_j \left(\widehat{v}_{j} - v_j \right) \right\| \nonumber 
\\ &\leq \sum_{j=1}^K \left| \widehat{\beta}_{j} - \beta_j \right| + \sum_{j=1}^K \left| \beta_j \right| \left\| \widehat{v}_{j} - v_j \right\|,
\end{align*}
by Minkowski's inequality and the fact that $\left\|\widehat{v}_{j} \right\| = 1$ for all $j = 1, \ldots, K$. The second term can be dealt with most easily by observing that by Parseval's equality for all $j$, $\left|\beta_j\right| \leq \max(1, \beta_j^2)\leq  \max(1, \sum_{j=1}^{\infty}\beta_j^2) = \max(1, \left\| \beta \right\|^2) \leq 1 + || \beta ||^2 < \infty$, since, by assumption, $\beta \in L^2 \left[0,1\right]$. Hence, the scores are uniformly bounded. Thus, the second term converges to zero in probability by the convergence in norm of the eigenfunctions. Lemma 4.2 covers the first term as an S-estimator may be treated as an M-estimator without updating the scale. Combining these two facts yields the result.

\end{proof}

\begin{customcor}{4.2}
Consider the vector $\widehat{\boldsymbol{\beta}} := \left( \widehat{\beta_0}, \widehat{\boldsymbol{\beta}}_1^{\top} \right)^{\top}$ from estimating equation (13). Under (C1)-(C4)
\begin{equation*}
\sqrt{n}(\widehat{\boldsymbol{\beta}} - \boldsymbol{\beta} ) \xrightarrow{D} \mathcal{N}\left(0,  \sigma^2 \frac{\mathbb{E}_F \left( \rho_{1}^{\prime}(\epsilon/ \sigma)^2 \right)}{(\mathbb{E}_F \left( \rho_1^{\prime \prime}(\epsilon/\sigma) \right))^2 } \diag\left(1, \lambda_1^{-1}, \ldots \lambda_K^{-1} \right) \right).
\end{equation*}
\end{customcor}

\begin{proof}
First, note that by (34) and (35) this is the asymptotic distribution of the "theoretical" MM-estimator $\widetilde{\boldsymbol{\beta}}$, see \citep{maronna2006robust}. But since
\begin{equation*}
\sqrt{n}(\widehat{\boldsymbol{\beta}} - \boldsymbol{\beta} ) = \sqrt{n}(\widehat{\boldsymbol{\beta}} - \widetilde{\boldsymbol{\beta}} ) + \sqrt{n}(\widetilde{\boldsymbol{\beta}} - \boldsymbol{\beta} ),
\end{equation*}
Slutzky's theorem would imply the result provided that $\sqrt{n}(\widehat{\boldsymbol{\beta}} - \widetilde{\boldsymbol{\beta}} ) \xrightarrow{P} 0 $. The previous representation of this difference in equation (\ref{eq:43}) implies that 
\begin{equation*}
\sqrt{n}(\widehat{\boldsymbol{\beta}} - \widetilde{\boldsymbol{\beta}} ) = \mathbf{A}_n \sqrt{n} ( \mathbf{f}(\widetilde{\boldsymbol{\beta}}) - \widetilde{\boldsymbol{\beta}}  ),
\end{equation*}
with $\mathbf{f}: \mathbb{R}^{K+1} \to \mathbb{R}^{K+1}$ the fixed point function and $\mathbf{A}_n$ a matrix whose entries are bounded in probability. The difference on the right tends to zero in probability because
\begin{equation*}
\frac{1}{\sqrt{n}}  \sum_{i=1}^n \left( \widehat{w}_i ( \boldsymbol{\widetilde{\beta}}) \mathbf{\widehat{x}}_i y_i - w_i ( \boldsymbol{\widetilde{\beta}}) \mathbf{x}_i y_i \right) = o_{P}(1),
\end{equation*}
which can be established by bounding the difference $\widetilde{\mathbf{x}}_i - \mathbf{x}_i $, as before, and then applying the Central Limit Theorem to show that terms such as $n^{-1/2} \sum_{i=1}^n |y_i| \left\|X_i - \mathbb{E}(X) \right\|$ are $O_{P}(1)$ while at the same time using the consistency of the eigenfunctions and the MM-estimators of location.

\end{proof}

\newpage
\bibliographystyle{apalike}
\bibliography{NP}

\end{document}